\documentclass[%
 superscriptaddress,
 amsmath,amssymb,
 aps,
 prx,
 twocolumn,
 nofootinbib
]{revtex4-2}
\usepackage{minted}
\usepackage{amsfonts,amsmath, amssymb,bm}
\usepackage[utf8]{inputenc}
\usepackage{hyperref}
\usepackage{qcircuit}
\usepackage{enumitem}
\usepackage{graphicx}
\usepackage{subcaption}
\usepackage{comment}
\usepackage[justification=raggedright,singlelinecheck=false]{caption}
\graphicspath{{./figures}}

\newcommand{\no}{\nonumber}

\usepackage{pgf}
\usepackage{xparse}
\usepackage{tikz}
\usetikzlibrary{arrows,arrows.meta,backgrounds,calc,cd,matrix,math,positioning}
\usepackage{braket,amsmath,amssymb,amsfonts,amsthm,mathtools,bm,comment}

\newtheorem{definition}{Definition}
\newtheorem{lemma}{Lemma}

\usepackage{dcolumn}
\usepackage{bm}

\makeatletter
\renewcommand{\@fnsymbol}[1]{%
  \ensuremath{%
    \ifcase#1
      \ast
    \or
      \dagger
    \or
      \ddagger
    \or
      \S
    \or
      \P
    \or
      \ast\ast
    \or
      \dagger\dagger
    \else
      \@ctrerr
    \fi
  }%
}
\makeatother

\allowdisplaybreaks

\begin{document}

\preprint{APS/123-QED}

\title{A Quantum Algorithm for Nonlinear Electromagnetic Fluid Dynamics via Koopman-von~Neumann Linearization}

\author{Hayato Higuchi$^{*}$}
\email{higuchi.hayato.007@gmail.com}
\affiliation{{\it International Research Center for Space and Planetary Environmental Science, Kyushu University, Motooka, Nishi-ku, Fukuoka 819-0395, Japan}}
\affiliation{{\it QunaSys Inc., Hakusan, Bunkyo-ku, Tokyo 113-0001, Japan}}

\author{Yuki Ito$^{*}$}
\email{yuki.itoh.osaka@gmail.com}
\affiliation{%
Graduate School of Engineering Science, The University of Osaka,\\
1-3 Machikaneyama, Toyonaka, Osaka 560-8531, Japan
}

\author{Kazuki Sakamoto}
\email{kazuki.sakamoto.osaka@gmail.com}
\affiliation{%
Graduate School of Engineering Science, The University of Osaka,\\
1-3 Machikaneyama, Toyonaka, Osaka 560-8531, Japan
}

\author{Keisuke Fujii}
\email{fujii.keisuke.es@osaka-u.ac.jp}
\affiliation{%
Graduate School of Engineering Science, The University of Osaka,\\
1-3 Machikaneyama, Toyonaka, Osaka 560-8531, Japan
}
\affiliation{%
Center for Quantum Information and Quantum Biology, The University of Osaka, 
 1-2 Machikaneyama, Osaka 560-0043, Japan
}
\affiliation{RIKEN Center for Quantum Computing (RQC), Hirosawa 2-1, Wako, Saitama 351-0198, Japan}

\author{Akimasa Yoshikawa}
\email{yoshikawa.akimasa.254@m.kyushu-u.ac.jp}
\affiliation{{\it International Research Center for Space and Planetary Environmental Science, Kyushu University, Motooka, Nishi-ku, Fukuoka 819-0395, Japan}}
\affiliation{{\it Faculty of Science, Kyushu University, Motooka, Nishi-ku, Fukuoka 819-0395, Japan}}

\begingroup
\renewcommand{\thefootnote}{$*$}
\footnotetext{The authors contributed equally to this work.}
\endgroup
\setcounter{footnote}{0}

\date{\today}

\begin{abstract}
To simulate plasma phenomena, large-scale computational resources have been employed in developing high-precision and high-resolution plasma simulations. 
One of the main obstacles in plasma simulations is the requirement of computational resources that scale polynomially with the number of spatial grids, which poses a significant challenge for large-scale modeling.
To address this issue, this study presents a quantum algorithm for simulating the nonlinear electromagnetic fluid dynamics that govern space plasmas.
We map it, by applying Koopman-von~Neumann linearization, to the Schr\"{o}dinger equation and evolve the system using Hamiltonian simulation via quantum singular value transformation.
Our algorithm scales $O \left(s N_x \, \mathrm{polylog} \left( N_x \right) T \right)$ in time complexity with $s$, $N_x$, and $T$ being the spatial dimension, the number of spatial grid points per dimension, and the evolution time, respectively.
Comparing the scaling $O \left( s N_x^s \left(T^{5/4}+T N_x\right) \right)$ for the classical method with the finite volume scheme, this algorithm achieves polynomial speedup in $N_x$.
The space complexity of this algorithm is exponentially reduced from $O\left( s N_x^s \right)$ to $O\left( s \, \mathrm{polylog} \left( N_x \right) \right)$.
Numerical experiments validate that accurate solutions are attainable with smaller $m$ than theoretically anticipated and with practical values of $m$ and $R$, underscoring the feasibility of the approach.
As a practical demonstration, the method accurately reproduces the Kelvin-Helmholtz instability, underscoring its capability to tackle more intricate nonlinear dynamics. 
These results suggest that quantum computing can offer a viable pathway to overcome the computational barriers of multiscale plasma modeling.
\end{abstract}

\maketitle

\section{Introduction}
Nonlinear plasma simulations are indispensable for predicting and elucidating a wide range of phenomena, including space weather~\cite{Eastwood2017,Fuselier2024}, astrophysical jets~\cite{Komissarov2021}, magnetic confinement fusion plasmas~\cite{Hoelzl2024}, and plasma-assisted semiconductor manufacturing~\cite{Versolato2024,Kim2024}.  
In conventional plasma simulation research, large-scale computations have been conducted using high-performance computing (HPC) with methods such as the finite volume schemes~\cite{Godunov1959,Patankar1980}. 
However, achieving higher accuracy in significantly larger-scale systems, which is required for numerous practical applications as in Ref.~\cite{Yang2025,Garcia2025}, is impractical even with state-of-the-art HPC resources.  
Therefore, fundamentally new computational approaches are required to realize both higher accuracy and scalability in plasma simulations~\cite{Yang2025,Garcia2025}.

Quantum computing provides a potential solution to this challenge.
As long-standing researches have investigated, quantum computers are expected to solve certain problems with exponentially smaller time and space resources than classical computers.
For example, solving linear systems of equations~\cite{Harrow2009, costa_optimal_2022, morales2024quantum} and simulating quantum dynamics~\cite{feynman1986quantum, kitaev2002classical, lloyd1996universal, low2017optimal, low2019hamiltonian, Gilyen2019} are proven to have exponential quantum advantage under the plausible complexity theoretic assumption.
Recently, quantum algorithms have been proposed for a variety of classical linear partial differential equations~(PDEs), 
e.g., the heat and wave equations.
These algorithms apply quantum linear systems algorithms~(QLSAs)~\cite{berry2014high, berry2017quantum, linden2022quantum, krovi2023improved, berry2024quantum, low2024quantum} and Hamiltonian simulation~\cite{Costa2019, jin2024quantum, Jin2023a, Jin2023b,Lubasch2025}, some of which suggest the potential for exponential quantum advantage.
For nonlinear dynamics, recent studies have proposed the quantum algorithms based on Carleman linearization~\cite{Carleman1932,Kowalski1991,Liu2021,krovi2023improved} or Koopman-von Neumann (KvN) linearization~\cite{Koopman1931,Neumann1932,Engel2021,Tanaka2024}, which map nonlinear systems to linear ones.
However, the classes of systems solvable by such methods remain limited, and the development of quantum algorithms applicable to practical problems, as well as the analysis of their computational overhead, has not yet been sufficiently explored.

In this study, we propose a quantum algorithm for electromagnetic fluid dynamics, which represents a fundamental class of nonlinear plasma simulations, and demonstrate its quantum advantage through a rigorous complexity analysis. 
Given an initial state,
our algorithm generates the time-evolved quantum state containing the desired solutions.
To end this,
we assume the plasma density is constant in time,
apply a variable transformation inspired by the energy conservation law, and employ a tailored discretization scheme.
In our algorithm, we successfully transform the target system of equations into the Schr\"{o}dinger equation using KvN linearization based on Ref.~\cite{Tanaka2024}.
While the resulting Schr\"{o}dinger equation is infinite dimensional to capture the nonlinearity, 
we truncate it to a finite dimensional system with a sufficiently large dimension to achieve the desired accuracy.
This enables us to simulate the dynamics on a quantum computer using quantum singular value transformation (QSVT)~\cite{Gilyen2019}.
Consequently, 
the proposed method reduces the time complexity from the classical finite volume scaling of $O\left(s N_x^{s}(T^{5/4}\epsilon^{-1/4}+T N_x) \right)$ 
to $O \left(s \, \mathrm{polylog} \left( N_x \right) \left( N_x T + \log \left(1/\epsilon \right) \right) \right)$, 
where $s$ is the spatial dimension,
$N_x$ is the number of grid points per direction, 
$T$ is the evolution time, 
and $\epsilon$ is the additive error between the output state and the exact time-evolved state.
Thus, we have a polynomial speedup with respect to $N_x$ 
and an exponential speedup with respect to $s$.
This method also improves the space complexity from $O(sN_x^s)$ to $O\left( s \, \mathrm{polylog} \left( N_x \right) \right)$.
This means that our method has an exponential space advantage over the classical methods.

In addition to the theoretical analysis above, 
we perform numerical simulations to quantify the truncation error achievable with the proposed algorithm.  
To investigate error behavior, we simulate 1D electromagnetic fluid dynamics with the proposed algorithm in the following three cases:
As the first case, we provide input that reproduces the linear plasma oscillation solution, known as simple harmonic motion, to examine the accuracy of the proposed algorithm for linear solutions by comparing it with the exact solution.
The result shows excellent agreement with the exact solution, maintaining high fidelity with a relative error below \(<10^{-2}\) over about $8$ oscillation periods.
In the second case, 
we conduct a nonlinear advection test with varying the KvN truncation order $m$, 
to investigate the algorithm’s capability to resolve nonlinear behavior, 
depending on the KvN truncation order.
This case shows that accuracy improves only when the KvN truncation order \(m\) and the truncation index of the QSVT polynomial approximation \(R\) are increased simultaneously; enlarging either parameter in isolation is insufficient.
The third case is a nonlinear advection test with varying numbers of grid points $N_x$, 
intended to assess the algorithm’s scalability with increasing system size.
We confirm that the scheme attains even higher accuracy as the number of grid points is increased.
This establishes that increasing the number of qubits ensures further higher accuracy.
Furthermore, to assess the algorithm’s performance on a practical nonlinear problem, we next simulate the 2D Kelvin–Helmholtz instability~\cite{Helmholtz1868,Thomson1871}.  
The Kelvin–Helmholtz instability serves as a rigorous benchmark test for accurately capturing multiscale and nonlinear physics across a wide range of domains~\cite{Faganello2017,Berlok2019,Wang2017}.
As benchmarks, both the nonlinear vortex roll up and the growth rate of the unstable mode agree closely with theoretical predictions, confirming that the proposed algorithm can faithfully reproduce the Kelvin–Helmholtz instability.

These numerical experiments validate that, in practical problem settings, sufficient accuracy can be attained with smaller than predicted values of $m$ and $R$.
This work provides one of the earliest demonstrations of a quantum algorithm applied to a practical nonlinear classical system governed by partial differential equations, supported by a set of numerical benchmarks for nonlinear dynamics.
These findings extend the practical scope of quantum computing and point toward a viable pathway for overcoming the computational bottlenecks that currently limit large-scale plasmas and other nonlinear fluid systems.

This paper is organized as follows.  
Sec.~\ref{sec:Preliminary} reviews the electromagnetic fluid dynamics, introduces the KvN linearization, 
and summarises the essentials of QSVT-based Hamiltonian simulation.  
Sec.~\ref{sec:Application} then details the application of KvN linearization to the target systems and provides a theoretical analysis of the resulting computational complexity and error bounds.  
Sec.~\ref{sec:Results} demonstrates the proposed approach on four test problems using a quantum circuit emulator, presenting the corresponding numerical results.  
Finally, Sec.~\ref{sec:Conclusion} is devoted to the conclusion and discussion.

\section{PRELIMINARY}\label{sec:Preliminary}
\subsection{Electromagnetic Fluid Dynamics}
Based on Refs.~\cite{Baumjohann1996,Goldston1995}, we introduce electromagnetic fluid dynamics as our target.
We define a plasma density~$n$, a fluid velocity~$\mathbf{u}$, an electric field~$\mathbf{E}$ and magnetic field~$\mathbf{B}$ in 3D physical space as
\begin{align}
n &= n(x,y,z),\\
\mathbf{u} &= (u_1(x, y, z, t), u_2(x, y, z ,t), u_3(x, y, z, t)),\\
\mathbf{E} &= (E_1(x, y, z, t), E_2(x, y, z ,t), E_3(x, y, z, t)),\\
\mathbf{B} &= (B_1(x, y, z, t), B_2(x, y, z, t), B_3(x, y, z, t)).
\end{align}
The coupled system of incompressible, pressureless fluid equations in an electromagnetic field and Maxwell's equations is given as:
\begin{align}
    0 &= \nabla \cdot (n\mathbf{u}),  \label{eq:conti} \\
    \frac{\partial \mathbf{u}}{\partial t}  &= -\left( \mathbf{u} \cdot \nabla \right) \mathbf{u} + \frac{q}{m_q} (\mathbf{E} + \mathbf{u} \times \mathbf{B}), \label{eq:MoE_fluid}\\
    \frac{\partial \mathbf{E}}{\partial t} &= \frac{1}{\epsilon_0 \mu_0} (\nabla \times \mathbf{B}) - \frac{q}{\epsilon_0} n\mathbf{u}, \label{eq:ampere}\\
    \frac{\partial \mathbf{B}}{\partial t} &= -(\nabla \times \mathbf{E}),\label{eq:farady}
\end{align}
where the plasma charge~\( q \), the plasma mass~\( m_q \), permittivity \( \epsilon_0 \) and permeability \( \mu_0 \).
Our target is this nonlinear system.
Note that, according to the continuity equation, Eq.~\eqref{eq:conti} implies a time‑stationary density, i.e., $\frac{\partial n}{\partial t}=0$.

The energy equation derived from Eq.~\eqref{eq:MoE_fluid} implies the following conservative system in the case of the periodic boundary:
\begin{equation}
    \begin{split}
        \frac{\partial}{\partial t} 
        \Biggl[ \int  &\frac{1}{2} m_q \left( \left\lvert \sqrt{n} \mathbf{u} \right\rvert^2
        + \left\lvert \sqrt{\frac{\epsilon_0}{m_q}} \mathbf{E} \right\rvert^2
        + \left\lvert \sqrt{\frac{1}{\mu_0 m_q}} \mathbf{B} \right\rvert^2 \right)  \\
        &dx dy dz \Biggr] = 0.
    \end{split}
    \label{eq:energy}
\end{equation}
This energy conservation later serves as a key in our study both for guiding the KvN linearization and as a primary metric in the error assessment.

\subsection{Koopman-von Neumann Linearization} \label{subsec:KvN-linearization}
Based on Refs.~\cite{kowalski_nonlinear_1997, Tanaka2024, ito_how_2023}, 
we introduce Koopman-von Neumann (KvN) linearization with the orthonormalized Hermite polynomial $\{H_n(x)\}_{n \in \mathbb{Z}_{\ge 0}}$.
We assume that the following differential equations are given,
\begin{equation} \label{eq:Nd-nonlinear-differential-eq}
    \frac{d}{dt} x_j = F_j (\bm{x}) \hspace{1em} (j \in \{0, \dots, N-1 \}),
\end{equation}
where $\bm{x}=(x_0, \dots, x_{N-1})^\top$, and analytic function $F \colon \mathbb{R}^N \ni \bm{x} \mapsto \left( F_0(\bm{x}), \dots, F_{N-1}(\bm{x}) \right)^\top \in \mathbb{R}^N$.

Let $\{ \ket{n_j} \}_{n_j=0}^\infty$ be number states of the particle number operator $\hat{N}_j$, 
satisfying $\hat{N}_j \ket{n_j}=n_j \ket{n_j}$.
Now we take the Hilbert space spanned by $\{ \otimes_{j=0}^{N-1} \ket{n_j} \}_{n_1, \dots, n_N \in \mathbb{Z}_{\ge 0}}$.
We define the position state $\ket{\bm{x}}=\otimes_{j=0}^{N-1} \ket{x_j} \ \left(x_0, \dots, x_{N-1} \in \mathbb{R} \right)$ as 
\begin{equation}
    \ket{\bm{x}}
    \coloneqq \bigotimes_{j=0}^{N-1} \left( w\left( x_j \right) \sum_{n_j=0}^\infty H_{n_j} (x_j) \ket{n_j} \right), \label{eq:quantum_state}
\end{equation}
where $w(x)=\exp(-x^2/2)$.
Let $\hat{x}_j \ (j \in \{0, \dots, N-1\})$ denote the position operator, 
which satisfies $\hat{x}_j \ket{\bm{x}}= x_j \ket{\bm{x}}$ for $x_0, \dots, x_{N-1} \in \mathbb{R}$.
We also introduce the momentum operator $\hat{k}_j \ (j \in \{0, \dots, N-1\})$,
satisfying $[ \hat{x}_j, \hat{k}_l ] = i \delta_{j, l}$ for $j,l \in \{0, \dots, N-1\}$.
To obtain the Schr\"{o}dinger equation of $\bm{x}(t)$, we assume that 
\begin{equation} \label{eq:div-0}
    \mathrm{div} \mathbf{F}(\bm{x}) = 0
\end{equation}
for $\bm{x} \in \mathbb{R}^N$.
Now, the solution $\bm{x}(t)$ of Eq. \eqref{eq:Nd-nonlinear-differential-eq} meets the differential equation
\begin{equation} \label{eq:koopman-definition}
    i \frac{d}{dt} \ket{\bm{x}(t)} = \hat{H} \ket{\bm{x}(t)},
\end{equation}
where the Hamiltonian $\hat{H}$ is
\begin{equation} \label{eq:hamiltonian-koopman-hermite-poly}
    \hat{H}=\sum_{j=0}^{N-1} \frac{1}{2} \left( \hat{k}_j F_j(\hat{x}_j) +  F_j(\hat{x}_j) \hat{k}_j \right).
\end{equation}
Therefore, $\ket{\bm{x}(T)}$ satisfies
\begin{align} \label{eq:koopman-definition-tr0}
    \ket{\bm{x}(T)} &= \exp(-i\hat{H}T)\ket{\bm{x}(0)}.
\end{align}

Next, 
following Ref.~\cite{Tanaka2024},
we introduce the ODE solvable with the KvN linearization.
\begin{definition}[Definition 2 in \cite{Tanaka2024}]\label{def:ODE-KvN}
    A system of $N$-dimensional ODEs in the interval $[0, T]$ solvable with the Koopman-von Neumann linearization 
    is a system of nonlinear ODEs defined by a family $\mathcal{S}$ of index sets 
    each of which determines the variables engaged in an interaction, 
    and coupling constants $\alpha_{p \rightarrow j}$ of the interactions:
    \begin{align} \label{eq:ODE-KvN}
        \frac{dx_j}{dt} = F_j(\bm{x}) = \sum_{p \in \mathcal{S} : j \in p} \alpha_{p \rightarrow j} \prod_{l \in p \setminus \{j\}} x_l \quad \left(j \in [N]\right),
    \end{align}
    where the index sets $p \in \mathcal{S}$, and the coupling constants $\alpha_{p \rightarrow j}$ satisfy the followings:
    \begin{enumerate}
        \item For any $p \in \mathcal{S}, \ 2 \leq |p| \leq d$.
        \item For any $j \in [N]$, there exists at least one and at most $c$ sets $p$ such as $j \in p$.
        \item For any $p = (j_1, \cdots, j_l) \in \mathcal{S}$, real numbers $\alpha_{p \rightarrow j_1}, \cdots, \alpha_{p \rightarrow j_l} \in \mathbb{R}$ are assigned such that $\sum_{j \in p} \alpha_{p \rightarrow j} = 0$.
    \end{enumerate}
\end{definition}
\noindent
Note that the form of Eq.~\eqref{eq:ODE-KvN} is constructed to satisfy Eq.~\eqref{eq:div-0}.
The conditions from 1 to 3 in Definition~\ref{def:ODE-KvN} are required to ensure the sparsity of the mapped Hamiltonian $\hat{H}$ in Eq.~\eqref{eq:hamiltonian-koopman-hermite-poly}.
The condition 3 in Definition~\ref{def:ODE-KvN} also ensures $\prod_{j=0}^{N-1} w\left( x_j \right)$ remains constant over time,
i.e., $\frac{\partial}{\partial t}\lVert \bm{x} \rVert^2 = 0$, 
and simplifies the time evolution of the amplitude.

Finally, we explain the truncation of the Hamiltonian in Eq.~\eqref{eq:hamiltonian-koopman-hermite-poly}, 
which restricts its action to a finite dimensional Hilbert space. 
Since a quantum computer can only operate on finite dimensional Hilbert space, 
this truncation enables the simulation of the dynamics governed by the truncated Hamiltonian on a quantum computer.
If we assume that up to $m$ particles are allowed in each mode, 
simulating Eq.~\eqref{eq:koopman-definition} requires $O(m^N)$ variables, 
which corresponds to $O(N \log m)$ qubits.
Since the number of qubits is proportional to the number of variables $N$, 
this method is not efficient in terms of the space complexity.
To achieve a logarithmic space overhead for a quantum computer with respect to $N$,
following Refs.~\cite{Engel2021, Tanaka2024},
we restrict the total particle number to at most $m$, 
which needs $O(N^m)$ variables.
This truncation leads to a reduction in the number of using qubits from $O(N \log m)$ to $O(m \log N)$,
which implies exponential reduction in the number of qubits with respect to the number of variables $N$.

\subsection{Hamiltonian simulation based on quantum singular value transformation} \label{subsec:implement-hamiltonian-simulation}

We introduce how to implement the time evolution operator based on quantum singular value transformation (QSVT)~\cite{Gilyen2019,Martyn2021}. 
QSVT is a quantum algorithm that applies polynomial transformations to the singular values of a given matrix.
Specifically, we can transform a Hamiltonian $\tilde{H}$ into a polynomial approximation of an exponential function as $P_R(\tilde{H})\approx \exp{(-i\tilde{H}T)}$,
where $R$ is the truncation index of Jacobi-Anger expansion.

Now, we describe the details of the Hamiltonian simulation using QSVT.
First, we construct a block encoding of Hamiltonian $\tilde{H}$, which encodes the Hamiltonian as the upper left block of a unitary matrix $U$.
That is, we implement $(\alpha, l, \epsilon)$-block encoding $U$ of $\tilde{H}$, which satisfies
\begin{equation}
    \left\lVert \tilde{H} - \alpha \left( \bra{0}^{\otimes l} \otimes I \right) U \left( \ket{0}^{\otimes l} \otimes I \right)  \right\rVert \le \epsilon.
\end{equation}

Next, we provide a polynomial approximation of $\exp(-ix\tau)$. 
Note that we can implement the time evolution operator $\exp(-i\tilde{H}T)$ using the substitutions of $x = \tilde{H}/\alpha$ and $\tau = \alpha T$.
Considering that
\begin{align}
    \exp\left(-ix\tau\right) = \cos\left(x\tau\right)-i\sin\left(x\tau \right),
\end{align}
it suffices to have polynomial approximations of $\cos \left(x \tau\right)$ and $\sin \left(x \tau\right)$.
To obtain this, we recall that
the Jacobi-Anger expansions of $\cos \left(x \tau\right)$ and $\sin \left(x \tau\right)$ are
\begin{align}
 & \cos \left(x \tau\right)=J_{0} (\tau)+2\sum _{k=1}^{\infty } (-1)^{k} J_{2k} (\tau)T_{2k} \left(x\right),\no\\
 & \sin \left(x \tau\right)=2\sum _{k=0}^{\infty } (-1)^{k} J_{2k+1} (\tau)T_{2k+1} \left(x \right),\label{eq:cos_sin_inf}
\end{align}
where $J_{m} (\tau)$ is the $m$-th Bessel function of the first kind and $T_{k} (x)$ is the $k$-th Chebyshev polynomial of the first kind.
Then we can use
the truncated series at an index $R$:
\begin{align}
P_{R}^{\cos}(x) & :=J_{0} (\tau)+2\sum _{k=1}^{R} (-1)^{k} J_{2k} (\tau)T_{2k} (x),\no\\
P_{R}^{\sin}(x) & :=2\sum _{k=0}^{R} (-1)^{k} J_{2k+1} (\tau)T_{2k+1} (x),
\label{eq:Difine_Pcos_Psin}
\end{align}
which are polynomial approximations.
According to Lemma 57 of Ref.~\cite{Gilyen2019}, the errors of this polynomial approximations are upper-bounded by
\begin{align}
\| \cos (x\tau)-P_{R}^{\cos}(x)\|  & \leq \frac{5}{4}\left(\frac{e|\tau|}{4( R+1)}\right)^{( 2R+2)} ,\no\\
\| \sin (x\tau)-P_{R}^{\sin}(x)\|  & \leq \frac{5}{4}\left(\frac{e|\tau|}{2( 2R+3)}\right)^{( 2R+3)} ,\label{eq:cos_sin_R}
\end{align}
for $x\in [-1,1]$.
Substituting $x= \tilde{H}/ \alpha$ and $\tau=\alpha T$, we obtain 
\begin{align}
&\left\Vert \exp (-i\tilde{H}T)-\left( P_{R}^{\cos}( \tilde{H}T) -iP_{R}^{\sin} ( \tilde{H}T)\right)\right\Vert \nonumber\\
&\quad\quad\leq \frac{5}{4}\left(\frac{e|\alpha T|}{4( R+1)}\right)^{( 2R+2)} \no\\
&\quad\quad\quad\quad+\frac{5}{4}\left(\frac{e|\alpha T |}{2( 2R+3)}\right)^{( 2R+3)}  
\label{eq:epsilon_QSVT} .
\end{align}
Note that the error decays exponentially in $R$, which leads to an efficiency of a quantum algorithm for Hamiltonian simulation.

According to Ref.~\cite{Gilyen2019} (see also Ref.~\cite{Toyoizumi2024}), 
by utilizing the $(\alpha,a,\epsilon/|2T|)$-block encoding $U$ of Hamiltonian $\tilde{H}$ and $P_{R}^{\cos}(x), P_{R}^{\sin}(x)$,
we can implement an $(1, a+2, \epsilon)$-block encoding $U_{\text{exp}}$ of $\frac{1}{2}\exp\left(-i \tilde{H} T\right)$ as shown in Fig.~\ref{fig:Uexp}: 
\begin{align}
    \left\Vert\frac{\exp\left(-i \tilde{H} T \right)}{2}-\langle 0|_{0,1,2} U_{\text{exp}} |0\rangle_{0,1,2}\right\Vert\leq \epsilon,
    \label{eq:QSVT_expH}
\end{align}
\noindent
where $|0\rangle_{0,1,2}$ means the $0,1,2$ ancilla qubit states corresponding to Fig.~\ref{fig:Uexp}.
The prefactor $1/2$ of $\exp(-i \tilde{H} T )$ comes from taking a linear combination of two operators $P_{R}^{\cos}( \tilde{H}T)$ and $P_{R}^{\sin} ( \tilde{H}T)$.
This prefactor can be canceled by oblivious amplitude amplification (OAA)~\cite{Gilyen2019} technique.
Using OAA, the total number of queries to the block encoding $U$ is given by $3\left(2R+1\right)$.

Eq.~\eqref{eq:epsilon_QSVT} gives the operator‑norm error that arises 
when the series is truncated at a finite degree $R$.
We have to note that our proposed scheme accumulates not only the KvN linearization error but also an additional error introduced by the QSVT procedure.
As $R$ increases, the error in Eq.~\eqref{eq:epsilon_QSVT} gets smaller. 
Moreover, the QSVT error plays a role analogous to the classical CFL condition and is therefore crucial when discussing subsequent numerical stability.

\begin{figure*}
\begin{center}
\rotatebox{270}{
\includegraphics[scale=0.7]{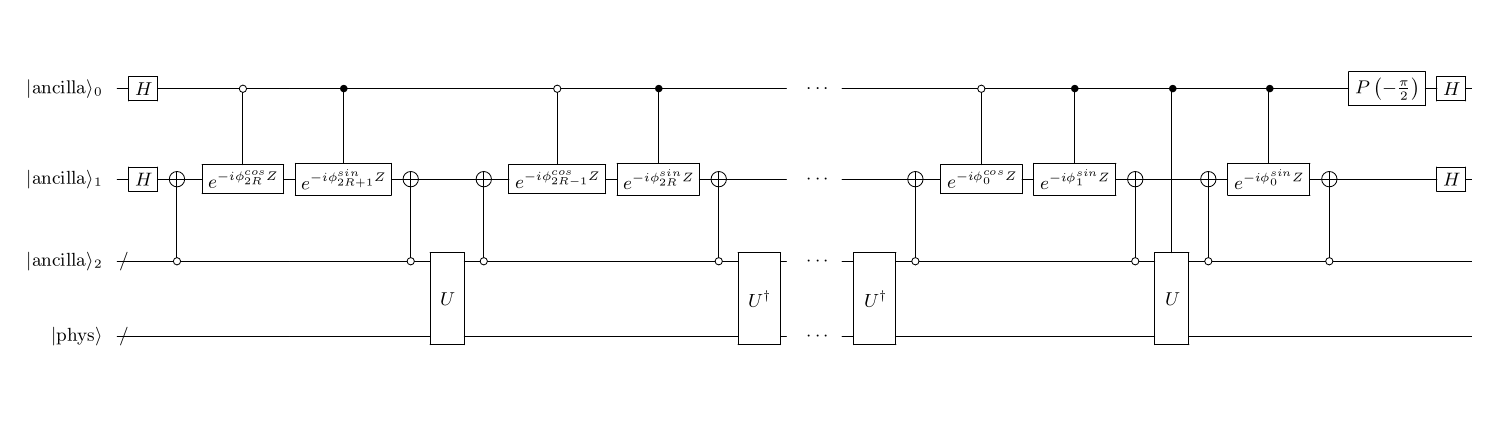}
}
\caption{The quantum circuit $U_{\text{exp}}$ for an $(1, a+2, \epsilon)$-block encoding of the time evolution operator $\frac{1}{2}\exp(-i\tilde{H}T)$. 
Here $U$ is an $(\alpha,a,\epsilon/|2T|)$-block encoding of Hamiltonian $\tilde{H}$.
See also Ref.~\cite{Toyoizumi2024}.}
\label{fig:Uexp}
\end{center}
\end{figure*}

\section{Quantum algorithm for electromagnetic fluid dynamics}\label{sec:Application}
In this section, we construct a quantum algorithm for electromagnetic fluid dynamics.
First, we discretize the target equations and 
obtain the mapped equations by KvN linearization.
Next, we provide the computational complexity of the proposed method, comparing that of a classical simulation.
We also construct the QSVT for Hamiltonian simulation of mapped Schr\"{o}dinger equation.
Finally, we evaluate the scaling of total error and discuss numerical stability.

\subsection{Mapping to the Schr\"{o}dinger equation}
We discretize the target equations and map them onto the Schr\"{o}dinger equation via the KvN linearization 
to simulate electromagnetic fluid dynamics on quantum computers.
First, 
referring to Eq.~\eqref{eq:energy},
we define $\tilde{\mathbf{u}}=\sqrt{n} \mathbf{u}$, $\tilde{\mathbf{E}}=\sqrt{\epsilon_0/m_q} \mathbf{E}$ and $\tilde{\mathbf{B}}=\sqrt{1/ (\mu_0 m_q)} \mathbf{B}$
such that we obtain suitable target equations for KvN linearization.
Then, the target Eqs.~\eqref{eq:MoE_fluid}, \eqref{eq:ampere} and \eqref{eq:farady} under Eq.~\eqref{eq:conti}
become
\begin{align}
    \frac{\partial \tilde{\mathbf{u}}}{\partial t}
    &= \frac{1}{4n\sqrt{n}}\left(\tilde{\mathbf{u}}\cdot \left(\nabla n\right) \right) \tilde{\mathbf{u}}
    - \frac{1}{2\sqrt{n}}\left(\nabla\cdot \tilde{\mathbf{u}} \right) \tilde{\mathbf{u}} \no
    \\
    & \quad - \frac{1}{\sqrt{n}}\left(\tilde{\mathbf{u}}\cdot \nabla \right) \tilde{\mathbf{u}} +\sqrt{\frac{n}{\epsilon_0m_q}}q \tilde{\mathbf{E}} 
    + \sqrt{\frac{\mu_0}{m_q}}q \tilde{\mathbf{u}} \times \tilde{\mathbf{B}}, \label{eq:modified_motion} \\
    \frac{\partial \tilde{\mathbf{E}}}{\partial t} 
    &= \frac{1}{\sqrt{\epsilon_0 \mu_0}} \left( \nabla \times \tilde{\mathbf{B}} \right) 
    -  \sqrt{\frac{n}{\epsilon_0m_q}} q\tilde{\mathbf{u}}, \label{eq:modified_ampere} \\
    \frac{\partial \tilde{\mathbf{B}}}{\partial t} 
    &= - \frac{1}{\sqrt{\epsilon_0 \mu_0}}\left( \nabla \times \tilde{\mathbf{E}} \right),\label{eq:modified_farady}
\end{align}
with periodic boundary condition.
Next, we discretize $\tilde{\mathbf{u}}, \tilde{\mathbf{E}}$ and $\tilde{\mathbf{B}}$ as
\begin{align}
    \begin{split}
        \tilde{\mathbf{u}}_{\mathbf{j}}(t) 
        &= \left( \tilde{u}_1(\mathbf{j} \Delta \mathbf{r}, t), 
        \tilde{u}_2(\mathbf{j} \Delta \mathbf{r}, t),
        \tilde{u}_3(\mathbf{j} \Delta \mathbf{r}, t)\right) \\
        &{\eqqcolon \left( \tilde{u}_{1,\mathbf{j}} (t), \tilde{u}_{2,\mathbf{j}} (t), \tilde{u}_{3,\mathbf{j}} (t)   \right)},
    \end{split} \\
    \begin{split}
        \tilde{\mathbf{E}}_{\mathbf{j}}(t) 
        &= \left( \tilde{E}_1(\mathbf{j} \Delta \mathbf{r}, t), 
        \tilde{E}_2(\mathbf{j} \Delta \mathbf{r}, t),
        \tilde{E}_3(\mathbf{j} \Delta \mathbf{r}, t)\right) \\
        &{\eqqcolon \left( \tilde{E}_{1,\mathbf{j}} (t), \tilde{E}_{2,\mathbf{j}} (t), \tilde{E}_{3,\mathbf{j}} (t)   \right)},
    \end{split} \\
    \begin{split}
        \tilde{\mathbf{B}}_{\mathbf{j}}(t) 
        &= \left( \tilde{B}_1(\mathbf{j} \Delta \mathbf{r}, t), 
        \tilde{B}_2(\mathbf{j} \Delta \mathbf{r}, t),
        \tilde{B}_3(\mathbf{j} \Delta \mathbf{r}, t)\right) \\
        &\eqqcolon {\left( \tilde{B}_{1,\mathbf{j}} (t), \tilde{B}_{2,\mathbf{j}} (t), \tilde{B}_{3,\mathbf{j}} (t)   \right)},
    \end{split}
\end{align}
where $\mathbf{j}=(j_1, j_2, j_3) \in \left\{0,1, \dots, N_x-1  \right\}^3$,
$\mathbf{j\Delta r}
=(j_1\Delta r_1, j_2\Delta r_2, j_3\Delta r_3) \in 
\{0, \Delta r_1, \dots, (N_x-1) \Delta r_1 \} \times 
\{0, \Delta r_2, \dots, (N_x-1) \Delta r_2 \} \times 
\{0, \Delta r_3, \dots, (N_x-1) \Delta r_3 \}$.
We assume $\Delta r_1, \Delta r_2, \Delta r_3 = O(1/N_x)$.
To satisfy periodic boundary condition, 
we define, for $\mathbf{j}^\prime=(j_1^\prime, j_2^\prime, j_3^\prime) \in \mathbb{Z}^3$, 
$\tilde{\mathbf{u}}_{\mathbf{j}^\prime} \coloneqq \tilde{\mathbf{u}}_{\mathbf{j}}$, 
where $\mathbf{j} \in \left\{0,1, \dots, N_x-1  \right\}^3$ is chosen such that 
$j_1 \equiv j_1^\prime \pmod{N_x}, \ j_2 \equiv j_2^\prime \pmod{N_x}, \  j_3 \equiv j_3^\prime \pmod{N_x}$.
Similarly, 
$\tilde{\mathbf{E}}_{\mathbf{j}^\prime}$ and $\tilde{\mathbf{B}}_{\mathbf{j}^\prime}$ are defined in the same way.
Now, we can represent 
the first, second, and third terms on the right-hand side of Eq.~\eqref{eq:modified_motion} as stated in Lemma~\ref{lem:nabla_3D}.

\begin{lemma} \label{lem:nabla_3D}
Using $\mathbf{r}=(r_1,r_2,r_3)$, 
$\Delta \mathbf{r}=(\Delta r_1,\Delta r_2,\Delta r_3)$, 
and $\mathbf{e}_1=(1, 0, 0), \mathbf{e}_2=(0, 1, 0), \mathbf{e}_3=(0,0,1)$,
we have
\begin{equation} \label{eq:nabla-3D}
    \begin{split}
        &\quad \frac{1}{4n\sqrt{n}}\left(\tilde{\mathbf{u}}\cdot \left(\nabla n\right) \right) \tilde{\mathbf{u}}
        - \frac{1}{2\sqrt{n}}\left(\nabla\cdot \tilde{\mathbf{u}} \right) \tilde{\mathbf{u}} 
        - \frac{1}{\sqrt{n}}\left(\tilde{\mathbf{u}}\cdot \nabla \right) \tilde{\mathbf{u}}\\
        &=\lim_{\Delta \mathbf{r} \to 0}
        \left\{ -
        \sum_{k=1}^3 
        \frac{1}{4 \Delta r_k}\left(\frac{1}{\sqrt{n_{\mathbf{j}+\mathbf{e}_k}}} 
        \tilde{u}_{k,\mathbf{j}
        + \mathbf{e}_k} \tilde{\mathbf{u}}_{\mathbf{j}+2\mathbf{e}_k} \right.\right.
        \\ &\quad\quad\quad\quad\quad\quad\quad\quad\quad\quad\left. \left.
        -\frac{1}{\sqrt{n_{\mathbf{j}-\mathbf{e}_k}}} 
        \tilde{u}_{k,\mathbf{j}-\mathbf{e}_k} \tilde{\mathbf{u}}_{\mathbf{j}-2\mathbf{e}_k}
        \right)
        \right\}.
    \end{split} 
\end{equation}
\end{lemma}

\begin{proof}
    See Appendix~\ref{appendix:lem-nabla_3D}.
\end{proof}

\noindent
Then, 
with the central difference and Lemma~\ref{lem:nabla_3D},
we can discretize Eqs.~\eqref{eq:modified_motion}, \eqref{eq:modified_ampere} and \eqref{eq:modified_farady} as follows:
\begin{widetext}
\begin{gather}
    \begin{split} \label{eq:target-u-discretize_3D}
        \frac{\partial \tilde{u}_{i,\mathbf{j}}}{\partial t}
        &=-
        \sum_{k=1}^3 
        \frac{1}{4 \Delta r_k}\left(\frac{1}{\sqrt{n_{\mathbf{j}+\mathbf{e}_k}}} 
        \tilde{u}_{k,\mathbf{j}
        + \mathbf{e}_k} \tilde{u}_{i,\mathbf{j}+2\mathbf{e}_k} 
        -\frac{1}{\sqrt{n_{\mathbf{j}-\mathbf{e}_k}}} 
        \tilde{u}_{k,\mathbf{j}-\mathbf{e}_k} \tilde{u}_{i,\mathbf{j}-2\mathbf{e}_k}
        \right)
        \\
        &\quad\quad+ \sqrt{\frac{n_{\mathbf{j}}}{\epsilon_0 m_q}} q \tilde{E}_{i,\mathbf{j}} 
        + \sqrt{\frac{\mu_0}{m_q}} q \sum_{k_2,k_3=1}^3\epsilon_{ik_2k_3} \tilde{u}_{k_2,\mathbf{j}} \tilde{B}_{k_3,\mathbf{j}},
    \end{split}  \\
    \begin{split} \label{eq:target-E-discretize_3D}
        \frac{\partial \tilde{E}_{i,\mathbf{j}}}{\partial t}
        &=\frac{1}{\sqrt{\epsilon_0 \mu_0}} 
        \sum_{k_2,k_3=1}^3 
        \epsilon_{ik_2k_3} \frac{1}{2 \Delta r_{k_2}}
        \left( 
        \tilde{B}_{k_3,\mathbf{j}+\mathbf{e}_{k_2}} 
        - \tilde{B}_{k_3,\mathbf{j}-\mathbf{e}_{k_2}}
        \right)
        - \sqrt{\frac{n_\mathbf{j}}{\epsilon_0 m_q}} q \tilde{u}_{i,\mathbf{j}} ,
    \end{split}  \\
    \begin{split} \label{eq:target-B-discretize_3D}
        \frac{\partial \tilde{B}_{i,\mathbf{j}}}{\partial t}
        &=-\frac{1}{\sqrt{\epsilon_0 \mu_0}} 
        \sum_{k_2,k_3=1}^3  
        \epsilon_{ik_2k_3} \frac{1}{2 \Delta r_{k_2}}
        \left(\tilde{E}_{k_3,\mathbf{j}+\mathbf{e}_{k_2}} 
        - \tilde{E}_{k_3,\mathbf{j}-\mathbf{e}_{k_2}}
        \right),
    \end{split} 
\end{gather}
\end{widetext}
where 
$i \in \{1,2,3\}$, 
$\mathbf{j}=(j_1,j_2,j_3) \in \{0, \dots, N_x-1\}^3$, and
the Levi-Civita symbol
\begin{equation}
    \epsilon_{k_1 k_2 k_3} =
    \begin{dcases}
        1 &\left( (k_1, k_2, k_3)=(1,2,3), (2, 3, 1), (3,1,2) \right) \\
        -1 &\left( (k_1, k_2, k_3)=(1,3,2), (2,1,3), (3,2,1) \right) \\
        0 &\left(\mathrm{otherwise} \right).
    \end{dcases}
\end{equation}

Next, to apply KvN linearization to Eqs.~\eqref{eq:target-u-discretize_3D}, \eqref{eq:target-E-discretize_3D}, and \eqref{eq:target-B-discretize_3D}, we confirm these equations satisfy the conditions in Definition~\ref{def:ODE-KvN}.
We represent the index sets $p=\left\{j_1, \dots, j_n \right\} \in \mathcal{S}$ and the coupling constants $\alpha_{p \rightarrow j} \ (j \in p)$ in Eq.~\eqref{eq:ODE-KvN}
as $\left\{ (\alpha_{p \to j_1}, x_{j_1}), \dots, (\alpha_{p \to j_n}, x_{j_n})  \right\}$.
Then, the index sets and the coupling constants of Eqs.~\eqref{eq:target-u-discretize_3D}, \eqref{eq:target-E-discretize_3D}, and \eqref{eq:target-B-discretize_3D} are as follows:
\begin{gather}
    \begin{split}
    & \left\{ \left(- \frac{1}{4 \Delta r_k\sqrt{n_{\mathbf{j}+\mathbf{e}_k}}}, \tilde{u}_{i,\mathbf{j}} \right), 
    \left(0, \tilde{u}_{k,\mathbf{j}+\mathbf{e}_k} \right),\right. \\ 
    &\hspace{3cm} \left. \left(\frac{1}{4 \Delta r_k\sqrt{n_{\mathbf{j}+\mathbf{e}_k}}}, \tilde{u}_{i, \mathbf{j}+2\mathbf{e}_k} \right)  \right\},
    \end{split}\\
    \left\{ \left(\sqrt{\frac{n_\mathbf{j}}{\epsilon_0 m_q}} q, \tilde{u}_{i, \mathbf{j}} \right), 
    \left(-\sqrt{\frac{n_\mathbf{j}}{\epsilon_0 m_q}} q, \tilde{E}_{i, \mathbf{j}} \right) \right\}
     , \\
    \left\{ \left(\sqrt{\frac{\mu_0}{m_q}} q, \tilde{u}_{1, \mathbf{j}} \right), 
    \left(-\sqrt{\frac{\mu_0}{m_q}} q, \tilde{u}_{2, \mathbf{j}} \right),
    \left(0, \tilde{B}_{3, \mathbf{j}} \right)  \right\}
     , \\
    \left\{ \left(-\sqrt{\frac{\mu_0}{m_q}} q, \tilde{u}_{1, \mathbf{j}} \right), 
    \left(\sqrt{\frac{\mu_0}{m_q}} q, \tilde{u}_{3, \mathbf{j}} \right),
    \left(0, \tilde{B}_{2, \mathbf{j}} \right)  \right\}
     , \\
    \left\{ \left(\sqrt{\frac{\mu_0}{m_q}} q, \tilde{u}_{2, \mathbf{j}} \right), 
    \left(-\sqrt{\frac{\mu_0}{m_q}} q, \tilde{u}_{3, \mathbf{j}} \right),
    \left(0, \tilde{B}_{1, \mathbf{j}} \right)  \right\}
     , \\
    \begin{split}
    &\left\{ \left(\frac{1}{\sqrt{\epsilon_0 \mu_0}} \frac{1}{2 \Delta r_{k_2}} \epsilon_{k_1 k_2 k_3}, \tilde{E}_{k_1, \mathbf{j}} \right), \right. \\
    &\hspace{1.5cm} \left.  \left(\frac{1}{\sqrt{\epsilon_0 \mu_0}} \frac{1}{2 \Delta r_{k_2}} \epsilon_{k_3 k_2 k_1}, \tilde{B}_{k_3,\mathbf{j}+\mathbf{e}_{k_2}} \right) \right\}, 
    \end{split}\\
    \begin{split}
    & \left\{ \left(-\frac{1}{\sqrt{\epsilon_0 \mu_0}} \frac{1}{2 \Delta r_{k_2}} \epsilon_{k_1 k_2 k_3}, \tilde{E}_{k_1, \mathbf{j}} \right), \right. \\
    &\hspace{1.5cm} \left.\left(-\frac{1}{\sqrt{\epsilon_0 \mu_0}} \frac{1}{2 \Delta r_{k_2}} \epsilon_{k_3 k_2 k_1}, \tilde{B}_{k_3,\mathbf{j}-\mathbf{e}_{k_2}} \right) \right\},
    \end{split}
\end{gather}
where $\mathbf{j} \in \{0, \dots, N_x-1\}^3, i,k,k_1,k_2,k_3 \in \{1, 2, 3 \}$.
These index sets and the coupling constants form Eqs.~\eqref{eq:target-u-discretize_3D}, \eqref{eq:target-E-discretize_3D} and \eqref{eq:target-B-discretize_3D} 
in the same way as Eq.~\eqref{eq:ODE-KvN}.
Furthermore, when $c=4, d=3$ in the 1D case, $c=8, d=3$ in the 2D case, and $c=14, d=3$ in the 3D case,
these index sets and the coupling constants satisfy the condition from 1 to 3 in Definition~\ref{def:ODE-KvN}.
Therefore, we can 
map Eqs.~\eqref{eq:target-u-discretize_3D}, \eqref{eq:target-E-discretize_3D} and \eqref{eq:target-B-discretize_3D} to a system of the ODEs defined in Definition~\ref{def:ODE-KvN}.

According to Ref.~\cite{Tanaka2024}, 
the mapped Hamiltonian
for a system of Eqs.~\eqref{eq:target-u-discretize_3D}, \eqref{eq:target-E-discretize_3D}, and \eqref{eq:target-B-discretize_3D}
is written as follows.
First, we rename the variables of Eqs.~\eqref{eq:target-u-discretize_3D}, \eqref{eq:target-E-discretize_3D}, and \eqref{eq:target-B-discretize_3D} to assign indices compatible with Eq.~\eqref{eq:ODE-KvN}.
\vspace{2mm}
\begin{itemize}
    \item 1D ElectroHydroDynamics
\end{itemize}
For $j \in \{0, \dots, 2N_x-1\}$ and $i=1$, 
let
\begin{equation}
    x_j \coloneqq
    \begin{dcases}
        \tilde{u}_{i,j} &(j=j_1), \\
        \tilde{E}_{i,j} &(j=N_x+j_1),
    \end{dcases}\label{eq:parameters_1D}
\end{equation}
and $N = 2N_x$.

\begin{itemize}
    \item 2D ElectroHydroDynamics
\end{itemize}
For $j \in \{0, \dots, 5N_x^2-1\}$ and $i \in \{1, 2\}$, 
let
\begin{equation}
    x_j \coloneqq
    \begin{dcases}
        \tilde{u}_{i,\mathbf{j}} &(j=(i-1)N_x^2 + j_1 N_x+j_2), \\
        \tilde{E}_{i,\mathbf{j}} &(j=2N_x^2+(i-1)N_x^2 +j_1 N_x+j_2), \\
        \tilde{B}_{3,\mathbf{j}} &(j=4N_x^2 +j_1 N_x+j_2),
    \end{dcases}\label{eq:parameters_2D}
\end{equation}
and $N = 5N_x^2$.

\begin{itemize}
    \item 3D ElectroHydroDynamics
\end{itemize}
For $j \in \{0, \dots, 9N_x^3-1\}$ and $i \in \{1,2,3\}$, 
let
\begin{equation}
    x_j \coloneqq
    \begin{dcases}
        \tilde{u}_{i,\mathbf{j}} &(j=(i-1)N_x^3 + j_1 N_x^2+j_2N_x+j_3), \\
        \tilde{E}_{i,\mathbf{j}} &(j=3N_x^3+(i-1)N_x^3 + j_1 N_x^2+j_2N_x+j_3), \\
        \tilde{B}_{i,\mathbf{j}} &(j=6N_x^3+(i-1)N_x^3 + j_1N_x^2+j_2N_x+j_3),
    \end{dcases}
\end{equation}
and $N = 9N_x^3$.
Then, the Hamiltonian $\hat{H}$ for the target equations is denoted as
\begin{equation} \label{eq:def_hamiltonian}
    \hat{H} \coloneqq
    \sum_{j=0}^{N-1} \sum_{p \in \mathcal{S} : j \in p} 
    \frac{\alpha_{p \rightarrow j}}{\Lambda^{\lvert p \rvert -2}} \hat{k}_j
    \prod_{l \in p \setminus \{j\}} \hat{x}_l,
\end{equation}
where $\Lambda (\ge 1)$ is a rescale parameter that reduces the nonlinearity of the given differential equations as its value increases.
Note that we have to truncate the Hamiltonian $\hat{H}$ in Eq.~\eqref{eq:def_hamiltonian} to execute the Hamiltonian simulation for a system of Eqs.~\eqref{eq:target-u-discretize_3D}, \eqref{eq:target-E-discretize_3D}, and \eqref{eq:target-B-discretize_3D} on quantum computers.
We denote $\hat{H}_m$ as the truncated Hamiltonian by the KvN truncation order $m$.

\subsection{The time and space complexity}

Next, we describe the time and space complexity 
of simulating the Schr\"{o}dinger equation for a system of Eqs.~\eqref{eq:target-u-discretize_3D}, \eqref{eq:target-E-discretize_3D}, and \eqref{eq:target-B-discretize_3D}.
We first analyze the complexities of the proposed quantum algorithm.
To realize the algorithm, 
we implement the time evolution operator $e^{-i \hat{H}_m T}$,
where $\hat{H}_m$ is the truncation of the Hamiltonian in Eq.~\eqref{eq:def_hamiltonian}.
The truncation procedure explained in Sec.~\ref{subsec:KvN-linearization} requires the relabeling of the indices of the state $\otimes_{j=0}^{N-1} \ket{n_j}$ to suppress the number of using qubits.
In this paper, 
we utilize the relabeling described in Appendix~\ref{appendix:relabeling}. 
To implement this time evolution operator $e^{-i \hat{H}_m T}$ in the quantum circuits, 
we construct the block encoding $U$ of the truncated Hamiltonian $\hat{H}_m$, 
as described in Sec.~\ref{subsec:implement-hamiltonian-simulation}.
This block encoding can be implemented by the method explained in Appendix~\ref{appendix:block-encoding-hamiltonian}.
According to Theorem 3 in Ref.~\cite{Tanaka2024}, 
we can use this block encoding $U$ to construct the time evolution operator $e^{-i \hat{H}_m T}$, 
as shown in Lemma~\ref{lem:HS-complexity}.

\begin{lemma} \label{lem:HS-complexity}
    For a given $(8 c m, \lceil m \log N \rceil+3, 0)$- block encoding $U$ of 
    $\tilde{H}=\hat{H}_m / \left( 3 \eta (m/2)^{3/2} \right)$ in Eq.~\eqref{eq:def_hamiltonian},
    where $c$ is defined in  Definition \ref{def:ODE-KvN} and
    $\eta = \max_{j \in p, p \in \mathcal{S}} \lvert \alpha_{p \to j} / \Lambda^{\lvert p \rvert -2} \rvert$, 
    we can implement $(1, \lceil m \log N \rceil+5, 6 \epsilon)$- block encoding of $e^{-i \hat{H}_m T}$, 
    with $O\left(\eta c m^{5/2} T + \log (1/\epsilon) \right)$ uses of $U$ or its inverse,
    3 uses of controlled-$U$ or its inverse and with $O\left( m (\log N) \left(\eta c m^{5/2} T + \log (1/\epsilon) \right)  \right)$ two-qubit gates and 
    using $O(1)$ ancilla qubits.
\end{lemma}

\begin{proof}
    We can take $d$ in Definition~\ref{def:ODE-KvN} as $d=3$.
    From Theorem 3 in Ref.~\cite{Tanaka2024}, the statement follows.
\end{proof}
\noindent

\noindent
Note that in the 1D case, we can set $c=4$; 
in the 2D case, we can set $c=8$; 
and in the 3D case, we can set $c=14$.
Also, remember that $\Delta r_1, \Delta r_2, \Delta r_3=O(1/N_x)$ and $\Lambda \ge 1$,
then $\eta$ becomes $\eta = O\left( N_x \right)$.
From Lemma~\ref{lem:HS-complexity},
we can implement $(1, \lceil m \log N \rceil+5, 6 \epsilon)$- block encoding of $e^{-i \hat{H}_m T}$ 
with $O\left(m^{5/2} N_x T + \log (1/\epsilon) \right)$ uses of $U$ or its inverse,
3 uses of controlled-$U$ or its inverse, $O\left( m s (\log N_x) \left( m^{5/2} N_x T + \log (1/\epsilon) \right)  \right)$ two-qubit gates and 
$O(1)$ ancilla qubits.
Here, the spatial dimension is denoted by $s$.
Referring to Sec.\ref{subsec:KvN-linearization},
this Hamiltonian simulation requires $O\left( m \log N \right)=O\left( ms \log N_x \right)$ qubits.

Next, we analyze the computational complexity of the finite volume method commonly used for classical electromagnetic fluid dynamics. 
A classical implementation requires $O(sN_x^s)$ memories to store the discretized vectors in the uniform grids.
We consider a classical simulation using the 1st-order upwind scheme in the spatial grids to avoid numerical oscillations and the 4th-order Runge-Kutta (RK4) scheme in time evolution. 
Determining the upwind direction requires a constant number of operations per grid point, leading to $O(1)$ operations per spatial derivative evaluation. 
The RK4 requires four evaluations of the spatial derivative per time step. 
To achieve a simulation time $T$ within an additive error $\epsilon$, the number of time evolutions for a $4$-th order method is typically $O(T^{5/4} \epsilon^{-1/4})$. 
In addition, the Courant-Friedrichs-Lewy~(CFL) condition requires $T/N_t=O(\Delta x)$, where $N_t$ is the number of time steps and $\Delta x \propto 1/N_x$. 
It suffices that $N_t = \Omega(TN_x)$.
Combining these considerations, the classical simulation up to time $T$ within the additive error $\epsilon$ requires $O(sN_x^s (T^{5/4}\epsilon^{-1/4}+TN_x))$ arithmetic operations.

Compared with the time and space complexity of the classical method,
those of the proposed quantum algorithm are reduced.
To analyze the complexities,
we assume $m=\mathrm{polylog}(N_x)$.
According to Ref.~\cite{Tanaka2024},
this assumption is valid to realize the target error $\epsilon$.
For time complexity,
our algorithm improves from $O\left(s N_x^{s}(T^{5/4}\epsilon^{-1/4}+T N_x) \right)$ to 
$O \left(s \, \mathrm{polylog} \left( N_x \right) \left( N_x T + \log \left(1/\epsilon \right) \right) \right)$,
which implies a polynomial speedup with respect to $N_x$
and an exponential speedup with respect to $s$.
Regarding space complexity,
the algorithm reduces from $O(sN_x^s)$ to $O\left( s \, \mathrm{polylog} \left( N_x \right) \right)$,
which means an exponential spatial advantage.

\subsection{Scaling of total error and numerical stability conditions}
To delineate the stability conditions of the proposed scheme, 
we analyze the combined error stemming from both the KvN linearization and the QSVT Hamiltonian simulation.  
As the total error approaches zero, 
the solution of the underlying unitary system becomes stable;
accordingly, we show the corresponding stability criterion.  
This criterion serves as a practical formula for choosing the truncation order and other key simulation parameters.

Let $\hat{H}_d$ be the KvN-linearized Hamiltonian, $\hat{H}_m$ be the KvN-truncated Hamiltonian, and $U_{\text{HS}}(:=2\langle \text{a}|_{0,1,2} U_{\text{exp}} |\text{a}\rangle_{0,1,2})$ be the QSVT-based Hamiltonian simulation operation based on Eq.~\eqref{eq:QSVT_expH}.
Applying the triangle inequality, we obtain
\begin{align}
    &\left|\langle c|\exp(-i\hat{H}_dT)|\bm{x}(0)\rangle -\langle c|U_{\text{HS}}|\bm{x}_m(0)\rangle \right| \no\\
    &\le \left|\langle c|\exp(-i\hat{H}_dT)|\bm{x}(0)\rangle-\langle c|\exp(-i\hat{H}_{m}T)|\bm{x}_m(0)\rangle\right| \no\\
    &\quad+\left| \langle c|\exp(-i\hat{H}_mT)|\bm{x}_m(0)\rangle-\langle c|U_{\text{HS}}|\bm{x}_m(0)\rangle\right|\\
    &\leq \left|\langle c|\exp(-i\hat{H}_dT)|\bm{x}(0)\rangle-\langle c|\exp(-i\hat{H}_{m}T)|\bm{x}_m(0)\rangle\right|\no\\
    &\quad+\|\exp(-i\hat{H}_mT)-U_{\text{HS}}\|,\label{eq:errors_(2)and(3)}
\end{align}
where $\langle c|$ is an arbitrary unit vector whose total particle number is 1 and $T$ is a simulation time.
In Eq.~\eqref{eq:errors_(2)and(3)}, 
the first term represents the KvN truncation error, 
while the second term corresponds to the QSVT-based Hamiltonian simulation error.

According to Ref.~\cite{Tanaka2024}, the KvN truncation error can be expressed as
\begin{align}
    &\quad \left|\langle c|\exp(-i\hat{H}_dT)|\bm{x}(0)\rangle-\langle c|\exp(-i\hat{H}_{m}T)|\bm{x}_m(0)\rangle\right|\no\\
    &\leq O\left(\left(\frac{CN_x}{\Lambda}\right)^{\lceil (m-1)/3 \rceil}\right)
    + \frac{2}{6^{\lceil (m-1)/3 \rceil} \cdot \left(\lceil (m-1)/3 \rceil ! \right)},
    \label{eq:error_KvN}
\end{align}
where $C$ is a constant that satisfies $\max_{j \in p, p \in \mathcal{S}} \lvert \alpha_{p \to j} \rvert \le C N_x$.
To hold Eq.~\eqref{eq:error_KvN}, 
we assume the rescale parameter $\Lambda$ meets $\Lambda \geq 2^{-1/2} \cdot 3^3 (CN_x)^2 c T m^3$.
According to Eq.~\eqref{eq:error_KvN}, the KvN‑linearization error converges to $0$ 
as \(m\) becomes sufficiently large only when the conditions \(\Lambda > CN_x\) and 
$\Lambda \geq 2^{-1/2} \cdot 3^3 (CN_x)^2 c T m^3$ are satisfied.

By substituting 
$\alpha = 2^{3/2} \cdot 3  C N_x c m^{5/2}$ into Eq.~\eqref{eq:epsilon_QSVT}, 
the QSVT‑based Hamiltonian simulation error can be written as
\begin{align}
&\left\lVert \exp(-i\hat{H}_mT)-U_{\text{HS}} \right\rVert \no\\
&\quad\quad\leq\frac{5}{4}\left(\frac{e|2^{3/2} \cdot 3 C N_x c m^{5/2} T|}{4( R+1)}\right)^{( 2R+2)}\quad& \no\\
&\quad\quad\quad\quad+\frac{5}{4}\left(\frac{e|2^{3/2} \cdot 3 C N_x c m^{5/2} T |}{2( 2R+3)}\right)^{( 2R+3)}. \label{eq:error_QSVT}
\end{align}
For sufficiently large \(R\), the QSVT‑based Hamiltonian‑simulation error remains bounded if the following stability condition is satisfied:
\begin{align}
\frac{ 3 eCN_xcm^{5/2}T}{2^{1/2}(R+1)}<1 .
\end{align}
By combining Eqs.~\eqref{eq:errors_(2)and(3)}, \eqref{eq:error_KvN}, and \eqref{eq:error_QSVT}, we obtain a unified estimate for the total error arising from KvN linearization and Hamiltonian simulation.  
Requiring this total error to vanish leads to the following stability condition:
\begin{align}
N_xT < \min{\left( \frac{2^{1/2}\Lambda}{3^3 C^2N_x c m^3}, 
\frac{2^{1/2}(R+1)}{3 eC cm^{5/2} } \right)},\label{eq:stability_KvN_QSVT}
\end{align}
which plays the role of a quantum analogue to the classical CFL condition~\cite{Courant1928}.  
This is because Eq.~\eqref{eq:stability_KvN_QSVT} corresponds to an upper limit for $N_xT$ similar to the Courant number.
Consequently, we should first choose $m$ to suppress the error in Eq.~\eqref{eq:error_KvN}, and then, based on this $m$, adjust $\Lambda$ or $R$ to satisfy Eq.~\eqref{eq:stability_KvN_QSVT}.

\section{Numerical Results}\label{sec:Results}
In this section, we numerically validate the quantum algorithm proposed for the electromagnetic fluid dynamics.  
The specific objectives of each test case are:
\begin{description}
  \item[Case (A)] To validate the proposed method in the linear solution, we compare its solution with the analytical result of 1D plasma oscillations.  
  \item[Case (B)] To investigate the behavior of the nonlinear numerical solution as the KvN truncation order $m$ varies, we perform 1D advection tests with different values of $m$.  
  \item[Case (C)] To demonstrate error convergence with increasing numbers of grid points $N_x$, we perform 1D advection tests by varying $N_x$. 
  \item[Case (D)] Reproduce the Kelvin–Helmholtz instability as a representative practical nonlinear problem.
\end{description}
In this study, we assume a sparse oracle for the Hamiltonian and directly implement the Hamiltonian matrix using the \texttt{SparseMatrix} function in Qulacs.
Then we implement our proposed quantum algorithm, denoted \textsf{KvN-QSVT}, using the quantum circuit emulator Qulacs~\cite{Suzuki2021}, following the approach in Ref.~\cite{Toyoizumi2024}.
For comparison, we also compute a reference solution, denoted \textsf{KvN‑expm}, by evaluating \(\exp(-i\hat{H}_{m}T)\) through direct diagonalization.
This results are available at Ref.~\cite{HiguchiItoCode2025}.

\subsection{1D linear electron plasma oscillation test}
To assess the accuracy of the linear solution produced by the proposed algorithm, 
we examine the 1D linear electron plasma oscillation, which is the most fundamental oscillatory mode in a 1D electromagnetic fluid system and one for which an exact solution exists. 
The objective of this test is to evaluate the accuracy of the proposed method by comparing its predicted oscillation frequency with the exact one.

The initial condition in Eq.~\eqref{eq:parameters_1D} is set as
\begin{align}
    u(x,t=0)=1, \quad E(x,t=0)=0,
\end{align}
where x-space range~$[-N_x/2,N_x/2]$, the number of spatial grids $N_x=8$, the spatial grid size~$\Delta x=1$, the number of parameters~$N=16$.
Then we set time steps~$N_t=200$, 
time parameter as a time step~$\tau=1$, 
plasma frequency~$\omega_{p.e}:=\sqrt{\frac{n_e}{\epsilon_0 m_e}} q_e=-1$, 
the rescale parameter $\Lambda=10^4$, 
degree of the polynomial approximation~$R=5$ and the KvN truncation order $m=1$.
The $L^2$-norm of the Hamiltonian is \( \alpha = 4.00 \), resulting in a modified real-time step of \( \Delta t = \tau / \alpha = 0.25 \). 

In the Case A, 
from Eqs.~\eqref{eq:quantum_state} and \eqref{eq:parameters_1D},
the implemented quantum state can be written as follows:
\begin{widetext}
\begin{align}
    \ket{\bm{x}(0)} &\coloneqq    
    \sum_{\substack{n_0, n_1, \dots, n_{N-1} = 0 \\ n_0 + n_1 + \dots + n_{N-1} \le m}}^m
    \prod_{j=0}^{N-1} 
    \left( w \left( \Lambda x_j(0) \right) H_{n_j} \left( \Lambda x_j(0) \right) \right) 
    \ket{n_0 n_1 \dots n_{N-1}}\\ 
    &= \bm{w} \left( \Lambda \bm{x}(0) \right) 
    \left(\lambda_{0}\ket{00\cdots 00} 
    + C_{1} \left( \tilde{u}_0\ket{10\cdots 00}
    + \hdots 
    + \tilde{E}_{N_x-1}\ket{00\cdots 01} \right) 
    + (\text{higher particle number terms}) \right), 
\end{align}
\end{widetext}
\noindent
where $\bm{w} ( \bm{x} ) \coloneqq \prod_{j=1}^N w\left( x_j \right)$,
$C_{1} \coloneqq \sqrt{2}\Lambda \pi^{N/4} /\|\mathbf{x}(0)\| $, and $\lambda_{0}$ are unphysical values with dimensional extension due to nonlinearity.

Figure~\ref{fig:u_E_1D_linear_QSVT_expm} shows the time evolution of $\tilde{u}$ and $\tilde{E}$ at $x=0$ obtained using \textsf{KvN‑QSVT} and \textsf{KvN‑expm}. 
Since the analytical solution of linear plasma oscillations is known to be a simple harmonic oscillation, the numerical results clearly exhibit the expected oscillatory behavior.
From the numerical results, the plasma oscillation period is determined as \( T_{\text{Numerical}} = 6.25 \). 
On the other hand, the exact plasma oscillation period obtained from linear analysis is \( T_{\text{Exact}} = 2\pi/\omega_{p,e}=6.28 \). 
The relative error between the numerical solution and the exact solution is $<10^{-2}$ over $7.96$ oscillation periods.
This result demonstrates that the numerical method reproduces the linear solution with high accuracy.
\begin{figure}
\begin{center}
\includegraphics[scale=0.3]{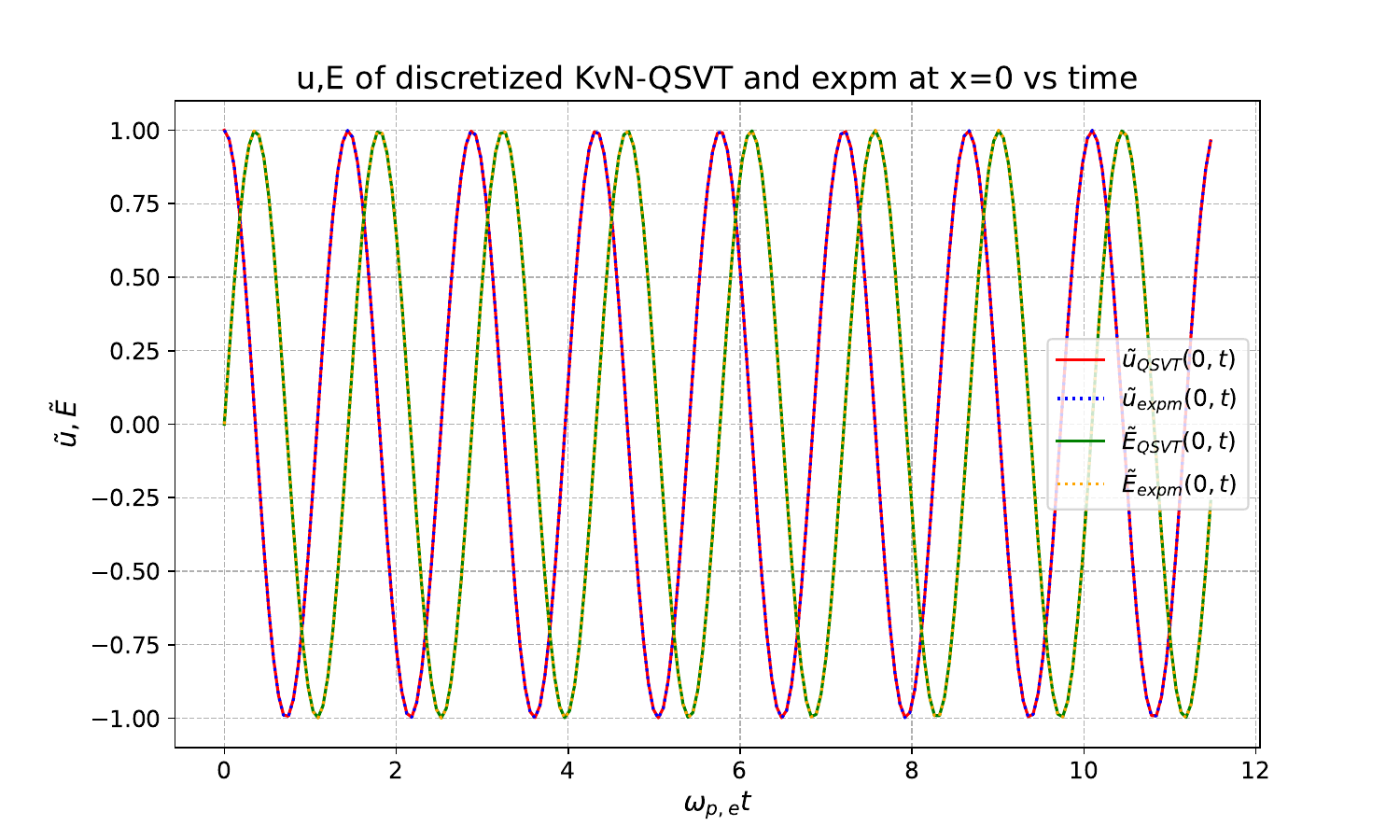}
\caption{Time evolution of the fluid velocity~$\tilde{u}$ and electric field~$\tilde{E}$ at a point of $x=0$. The solid lines are the discretized \textsf{KvN‑QSVT} result and the dot lines are the discretized \textsf{KvN‑expm}. Time scales with plasma frequency~$\omega_{p.e}$}
\label{fig:u_E_1D_linear_QSVT_expm}
\end{center}
\end{figure}

\subsection{1D nonlinear advection test with various KvN truncation order~$m$}
To investigate how the nonlinear solution depends on the KvN truncation order \(m\), 
we perform a 1D advection test.  
As an accuracy metric, 
we exploit the fact that the 1D advection equation conserves the $L^2$ norm of the $\mathbf{x}$, \(\|\mathbf{x}\|_{L^2}\) since the law of conservation of energy of the system Eq.~\eqref{eq:energy} holds. 
The goal is therefore to assess how closely the numerical solution preserves this invariant;
specifically, we monitor the quantity $\Delta(t)\;=\;\left|\bigl\|\mathbf{x}(t)\bigr\|_{L^2}\;-\;\bigl\|\mathbf{x}(0)\bigr\|_{L^2}\right|$,
and verify that  \(\Delta(t)\) approaches zero as the KvN truncation order $m$ increases, thereby indicating improved preservation of the invariant.
For different values of $m$, we quantify the truncation‑induced error in the proposed scheme.

The initial condition in Eq.~\eqref{eq:parameters_1D} is set as
\begin{align}
    u(x,t=0)=\sin(kx), \quad E(x,t=0)=0,
\end{align}
where x-space range~$[-N_x/2,N_x/2]$, the number of spatial grids~$N_x=8$, the spatial grid size~$\Delta x=1$, the initial wave number~$k=-2\pi/N_x$, the number of parameters~$N=16$.
 Then we set the plasma frequency~$\sqrt{\frac{n_e}{\epsilon_0 m_e}} q_e=-0.1$, 
the rescale parameter $\Lambda=1$, 
degree of the polynomial approximation~$R=5$,
time parameter as a time step~$\tau=1$, 
and time steps~$N_t=210,710,\text{ or }2056$ corresponding to various the KvN truncation order~$m=2,3,\text{ or }4$.
Real time range~$[0,100]$ is the same in all $m$ cases.

For several KvN truncation orders~\(m\),
Fig.~\ref{fig:L2norm_1D_nonlinear_deviation_error_various_m} shows the time evolution of the global $L^2$-norm error~\(\Delta(t)\) obtained with the \textsf{KvN-QSVT} scheme and with the reference \textsf{KvN-expm} solution. 
The \textsf{KvN-expm} results approach the norm of the initial solution as $m$ increases, which is the desirable behavior, whereas the \textsf{KvN-QSVT} curves display no such trend.
This discrepancy is attributed to the low-order polynomial approximation error introduced by QSVT.  
Since KvN linearization lifts the nonlinear system into a higher‑dimensional Hilbert space, 
the QSVT error acts on both the physical subspace and the augmented (Hilbert) subspace; 
when the influence of nonlinear terms is present, 
these two error contributions interact in a complicated manner, breaking norm conservation.  
Hence, to improve the accuracy of the proposed method, it is essential to increase not only \(m\) but also the truncation index of the QSVT polynomial approximation \(R\).

\begin{figure}
    \centering
    \includegraphics[scale=0.3]{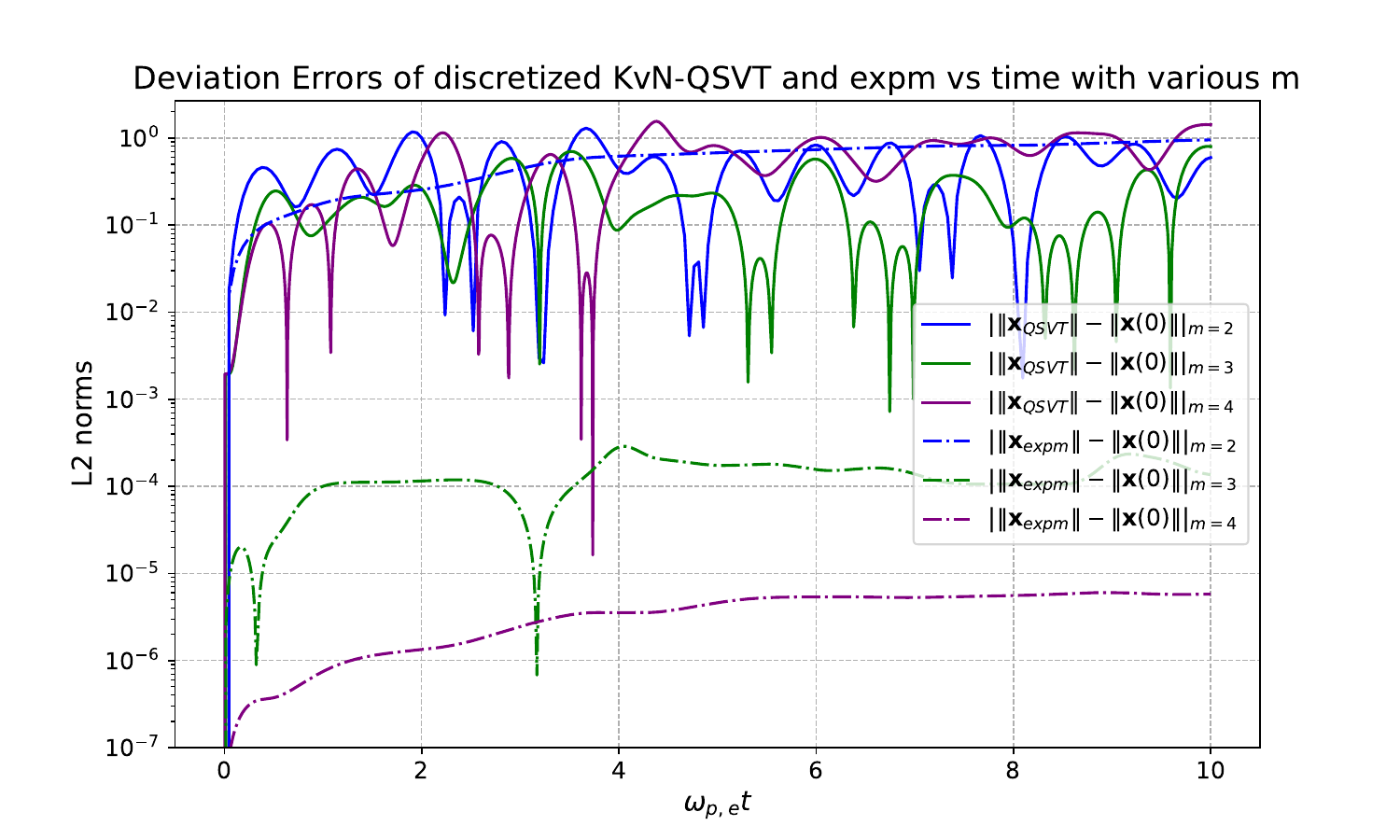}
    \caption{Time evolution of the deviation errors between the $L^2$ norm of variable vector $\mathbf{x}$ and the initial $L^2$ norm by the \textsf{KvN‑QSVT} and \textsf{KvN‑expm} in the Case B.}
\label{fig:L2norm_1D_nonlinear_deviation_error_various_m}
\end{figure}

\subsection{1D nonlinear advection test with various number of grids~$N_x$}
To verify that the proposed scheme attains higher accuracy as the grid resolution is increased, we perform a 1D advection test on a sequence of grids, where the resolution is controlled by the number of grid points \(N_x\). 
Demonstrating scalability with respect to \(N_x\) is a primary motivation for employing a quantum algorithm. 
The accuracy metric is the \(L^2\)-norm deviation 
$\Delta(t)\;=\; \left\lvert \bigl\|\mathbf{x}(t)\bigr\|_{L^2}\;-\;\bigl\|\mathbf{x}(0)\bigr\|_{L^2} \right\rvert$.
The goal is to show that \(\Delta(t)\to 0\) as \(N_x\) increases, as in Case B, 
so that the system satisfies the law of conservation of energy Eq.~\eqref{eq:energy}.

The initial condition in Eq.~\eqref{eq:parameters_1D} is set as
\begin{align}
    u(x,t=0)=\sin(kx), \quad E(x,t=0)=0,
\end{align}
where x-space range~$[-N_x/2,N_x/2]$, the various number of spatial grids~$N_x=11,22,33,\text{ or }44$, the spatial grid size~$\Delta x=1$, the initial wave number~$k=-2\pi/N_x$, the number of parameters~$N=2N_x$.
Then we set the plasma frequency~$\sqrt{\frac{n_e}{\epsilon_0 m_e}} q_e=-0.1$, 
$\Lambda=1$, 
degree of the polynomial approximation~$R=5$,
time parameter as a time step~$\tau=1$, 
the KvN truncation order~$m=2$, 
and time steps~$N_t=142,262,381, \text{ or }500$ corresponding to various $N_x=11,22,33,\text{ or }44$.
So, real time range~$[0,54.1]$ is the same in all $N_x$ cases.

Figure~\ref{fig:L2norm_1D_nonlinear_QSVT_expm_various_Nx} shows the temporal time evolution of the global $L^2$-norm error $\Delta(t)$ for several grid resolutions $N_x$, obtained using the \textsf{KvN-QSVT} scheme and compared with the reference \textsf{KvN-expm} solution.
The \textsf{KvN-expm} curves exhibit near-monotonic convergence as $N_x$ is increased, which is clearly visible in the figure. 
Although the \textsf{KvN-QSVT} results tend toward zero more slowly, their \(\Delta(t)\) values also trend toward zero, indicating improving accuracy.  
For both schemes, the improvement with increasing \(N_x\) arises because the discretization error decreases as \(O(1/N_x)\) according to Eq.~\eqref{eq:discretize_ududx}.  
This implies that \textsf{KvN-QSVT} benefits from grid refinement is a key property that guarantees the scalability of the proposed method, and this test confirms that behaviour.

\begin{figure}
    \centering
        \includegraphics[scale=0.3]{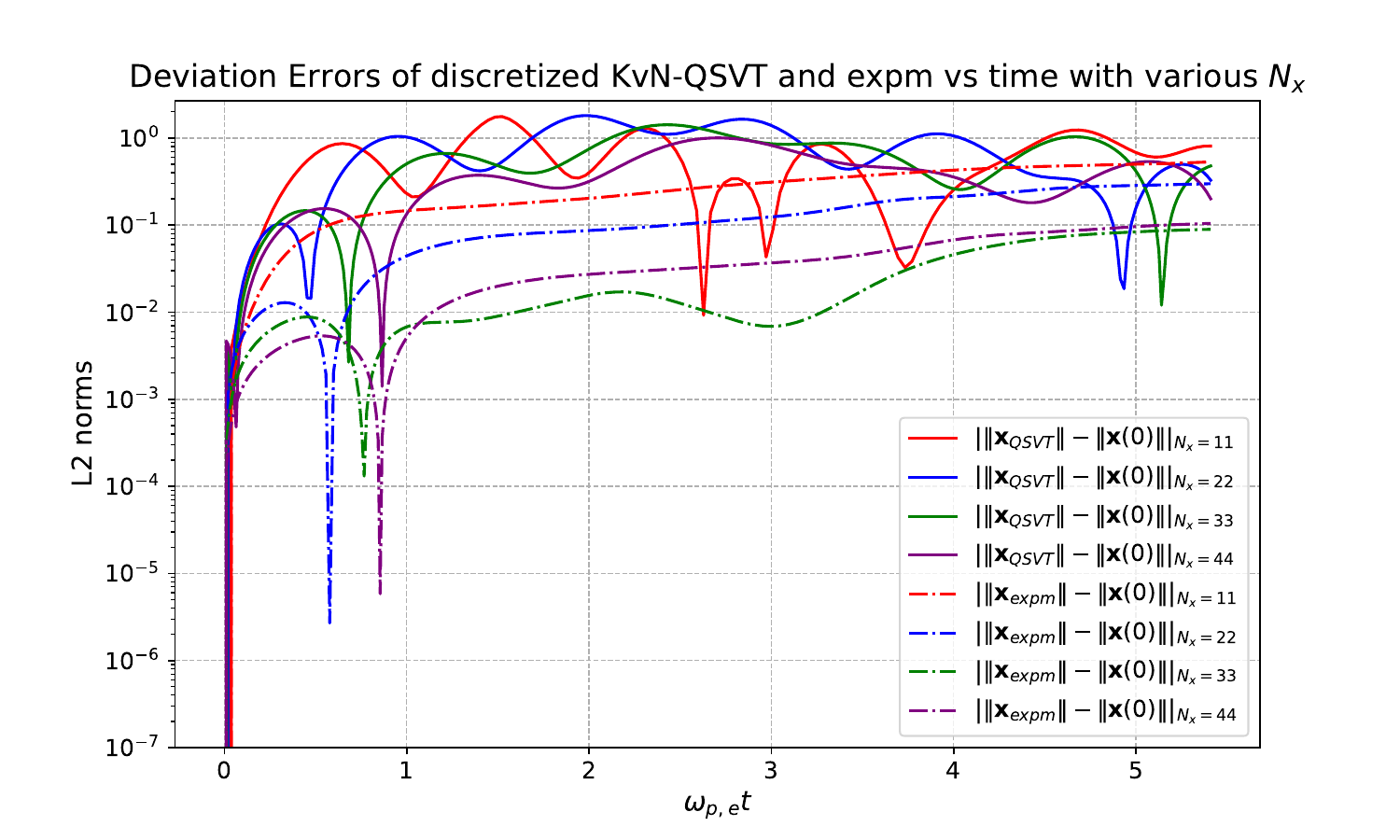}
    \caption{Time evolution of the deviation errors between the $L^2$ norm of variable vector $\mathbf{x}$ and the initial $L^2$ norm by the \textsf{KvN‑QSVT} and \textsf{KvN‑expm} in the Case C.}
\label{fig:L2norm_1D_nonlinear_QSVT_expm_various_Nx}
\end{figure}

\subsection{2D nonlinear Kelvin-Helmholtz instability test}
The Kelvin–Helmholtz instability is one of the most ubiquitous fluid phenomena in nature.  
Here, we reproduce the 2D Kelvin–Helmholtz instability as a representative nonlinear benchmark.  
The goals of this test are two‑fold:  
(i) to confirm the formation of vortices along the velocity shear layer, and  
(ii) to demonstrate that the growth rate of the initial unstable mode agrees with the theoretical value predicted by linear eigenvalue analysis.

The initial condition in Eq.~\eqref{eq:parameters_2D} is set as
\begin{align}
    u_1(x,y,t=0) &=
    \begin{cases}
        u_0 + \varepsilon \sin(k_x x + k_y y), \\ 
        \quad\quad\quad(-3\Delta y \leq y \leq 3\Delta y) \\
        -u_0,  \quad\quad\quad(\text{otherwise})
    \end{cases}\\
    u_2(x,y,t=0)&=0, \quad B_3(x,y,t=0)=B_0,\\
    E_1(x,y,t=0)&=0, \quad E_2(x,y,t=0)=0,
\end{align}
where x-space range~$[-N_x/2,N_x/2]$, y-space range~$[-N_y/2,N_y/2]$, the various number of spatial grids~$N_x=N_y=20$, the spatial grid size~$\Delta x=\Delta y=1$, the number of parameters~$N=2000$.
Then we set the plasma frequency~$\sqrt{\frac{n_e}{\epsilon_0 m_e}} q_e=-1$, 
$\Lambda=1$, 
degree of the polynomial approximation~$R=5$,
time parameter as a time step~$\tau=25$, 
the KvN truncation order~$m=2$, 
time steps~$N_t=2000$, 
and real time range~$[0,22.8]$.
To identify the mode with the maximum growth rate, we next perform an eigenvalue analysis of Eqs.~\eqref{eq:conti},\eqref{eq:MoE_fluid},\eqref{eq:ampere}, and \eqref{eq:farady}. 
We use the set of parameters that corresponds to the theoretically maximum growth rate of the unstable mode: 
the initial wave number~$k_x=0.99,k_y=0.2$, the initial background velocity~$u_0=1$, the initial magnetic field~$B_0=2$, and the disturbance parameter~$\varepsilon=0.1$.

Figure~\ref{fig:2DKH_velocity_field} shows the time evolution of the 2D velocity vector field obtained using \textsf{KvN‑QSVT}, 
which represents (a) $t=0.00$, (b) $t=1.07$, and (c) $t=2.15$.
The \textsf{KvN‑QSVT} results indicate the formation of velocity vortices at the velocity limits ($y=\pm3$) with exhibiting plasma oscillations (completed goal (i)).

Figure~\ref{fig:growth_rates} compares the growth rate of the unstable mode of the perturbation component $\tilde{u}_2$. 
We compare the results obtained using \textsf{KvN‑expm} (green line) with those from \textsf{KvN‑QSVT} (blue line). 
The red dashed line represents the theoretical unstable mode growth rate obtained from eigenvalue analysis corresponding to the initial perturbation setting. 
In Fig.~\ref{fig:growth_rates}, both \textsf{KvN‑QSVT} and \textsf{KvN‑expm} exhibit growth rates that agree with the theoretical prediction in the initial stage, thereby fulfilling goal (ii).
This suggests that \textsf{KvN‑QSVT} and \textsf{KvN‑expm} can reproduce the fundamental characteristics of Kelvin-Helmholtz instability.

\begin{figure*}
    \centering
    \begin{subfigure}{0.3\textwidth}
        \centering
        \includegraphics[scale=0.3]{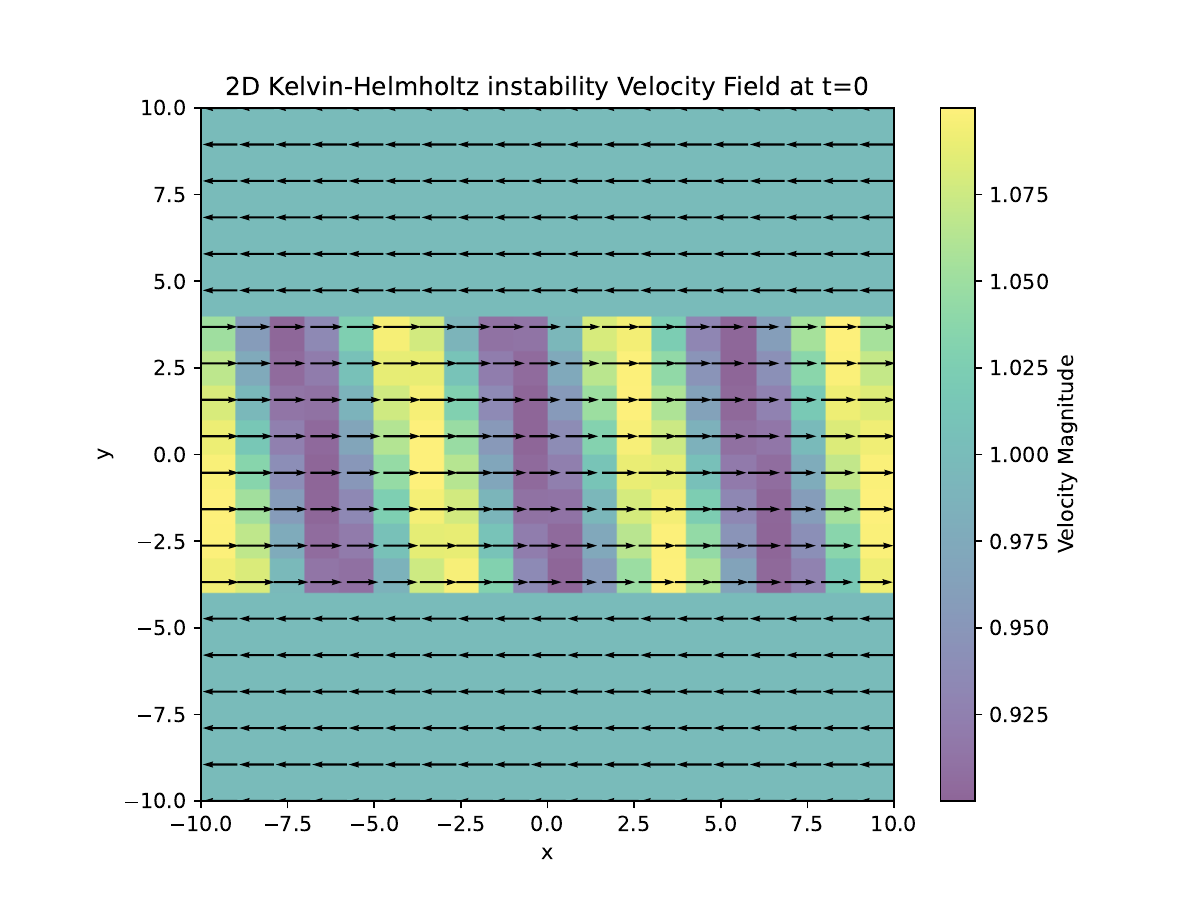}
        \caption{$t=0.00$}
    \end{subfigure}
    \begin{subfigure}{0.3\textwidth}
        \centering
        \includegraphics[scale=0.3]{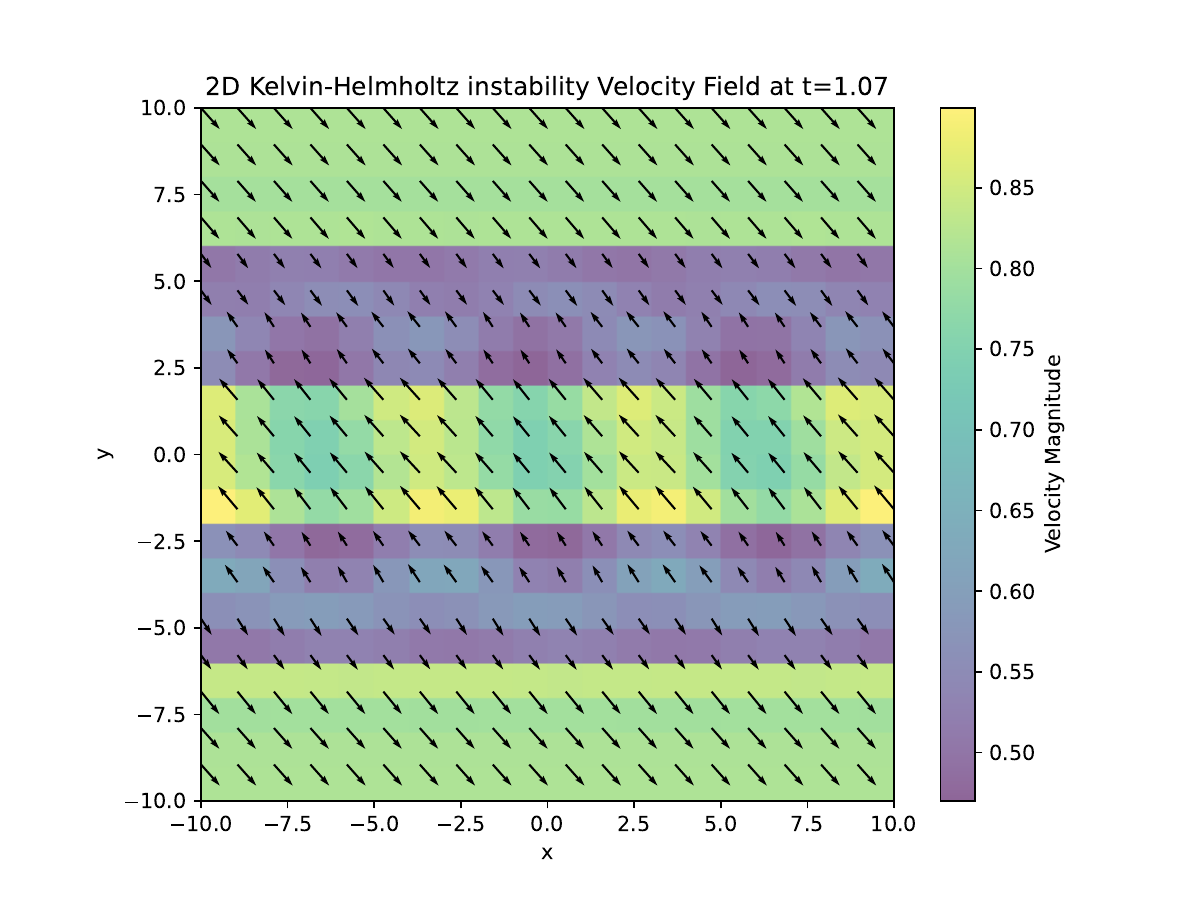}
        \caption{$t=1.07$}
    \end{subfigure}
    \begin{subfigure}{0.3\textwidth}
        \centering
        \includegraphics[scale=0.3]{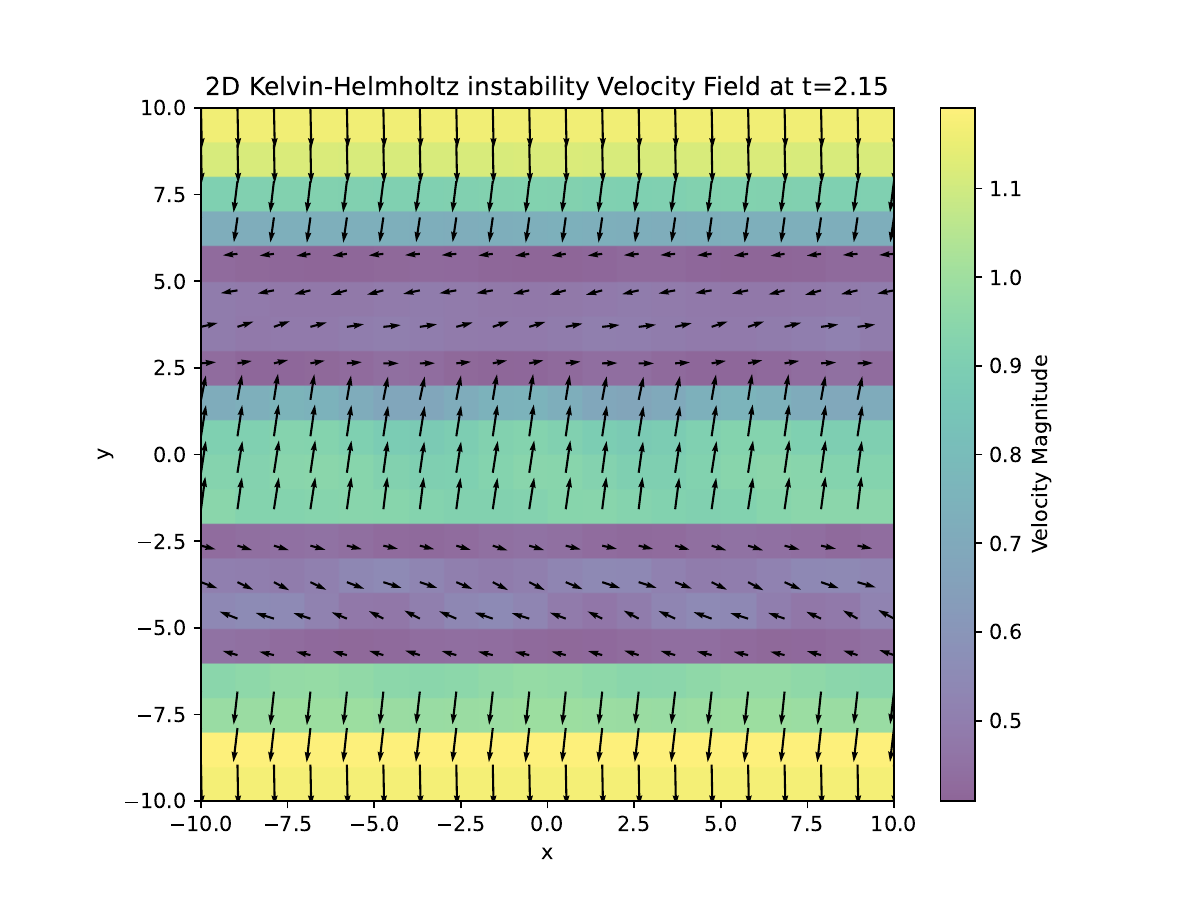}
        \caption{$t=2.15$}
    \end{subfigure}
    \caption{The velocity vector fields of the 2D Kelvin-Helmholtz instability obtained from \textsf{KvN‑QSVT} simulations. (a): \( t = 0 \), (b): \( t = 1.07 \) ($=0.5T_{\text{eign}}$), and (c): \( t = 2.15 \) ($=1.0T_{\text{eign}}$). The color bar represents the magnitude of the velocity vector.}
\label{fig:2DKH_velocity_field}
\end{figure*}

\begin{figure}
    \centering
    \includegraphics[scale=0.34]{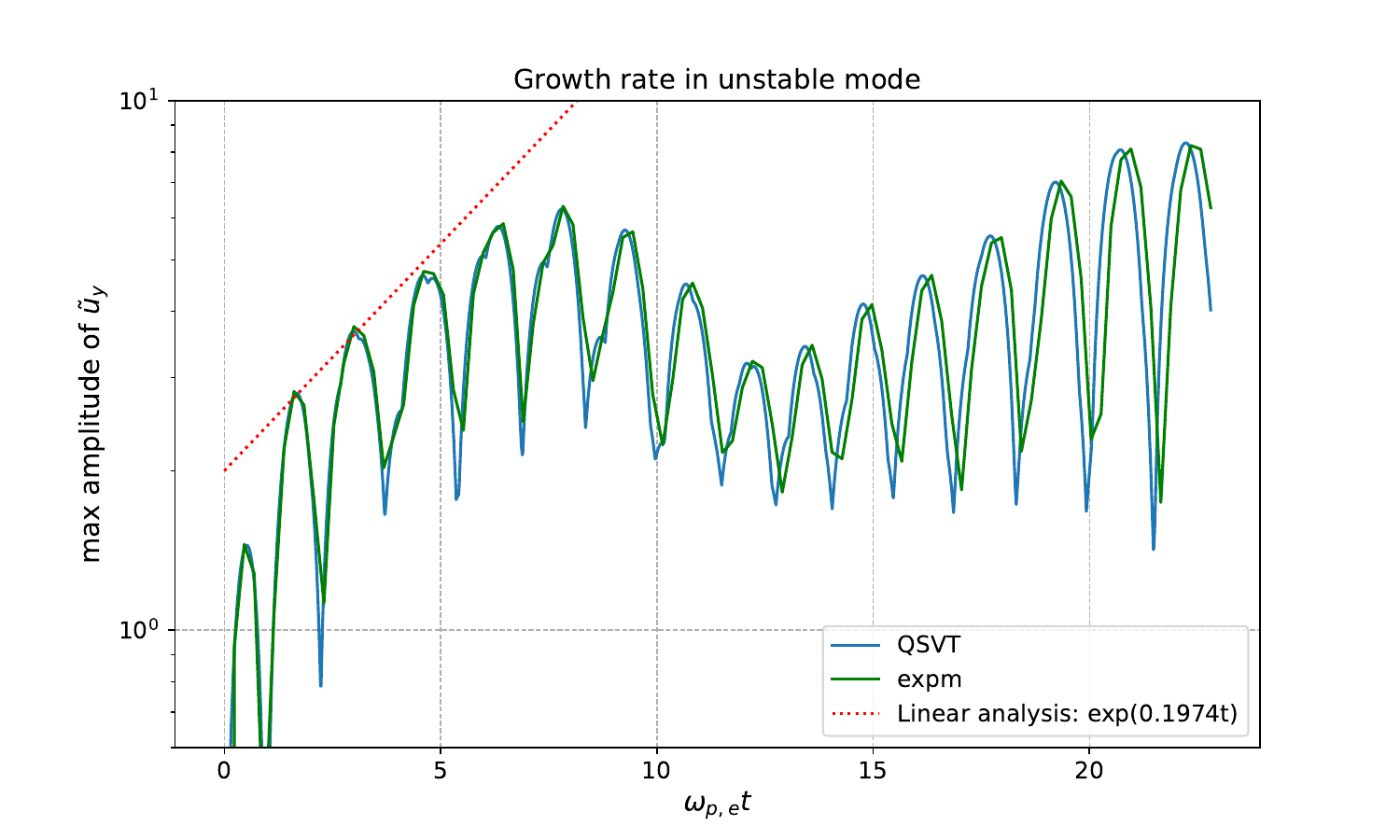}
    \caption{The time evolution of the maximum perturbation velocity \( \tilde
    {u}_y \) obtained from \textsf{KvN‑QSVT} (blue line) and \textsf{KvN‑expm} (green line). The red dashed line represents the growth rate of the unstable mode derived from linear analysis.}
\label{fig:growth_rates}
\end{figure}

\section{Conclusion and Discussion}\label{sec:Conclusion}
This study focuses on the fundamental equations of plasma dynamics and the application of quantum algorithms to nonlinear partial differential equations (PDEs).  
We propose a novel approach in which the electromagnetic fluid dynamics are linearized 
using the KvN formalism and the time evolution via QSVT.  
Utilizing the quantum circuit simulator Qulacs, 
we implement this method and demonstrate its performance across four numerical test cases.

Plasma fluid simulations are a representative application of electromagnetic fluid dynamics, 
yet large‑scale, high resolution computations of multidimensional systems are hindered by the $O\left(s N_x^{s}(T^{5/4}\epsilon^{-1/4}+T N_x) \right)$ scaling of classical algorithms in time complexity. 
The quantum scheme proposed here reduces this time complexity to $O \left(s \, \mathrm{polylog} \left( N_x \right) \left( N_x T + \log \left(1/\epsilon \right) \right) \right)$, 
which implies a polynomial speedup with respect to $N_x$ and an exponential speedup with respect to $s$.
The proposed algorithm also improves the space complexity from $O(sN_x^s)$ to $O\left( s \, \mathrm{polylog} \left( N_x \right) \right)$, 
which means an exponential space advantage over the classical methods.
Numerical experiments, using performance metrics such as energy conservation and instability growth rates, demonstrate that sufficient accuracy is achieved with values of $m$ and $R$ smaller than theoretically predicted in practical problem settings.
Hence, given adequate quantum resources, the method promises substantial quantum acceleration for large‑scale simulations.

\begin{acknowledgments}
Mr. K. Toyoizumi kindly provided the foundational code on which our QSVT-based Hamiltonian simulation implementation is built, for which we are deeply grateful.
Dr. J. W. Pedersen is acknowledged for engaging in insightful discussions with HH that served as the seed of this study.
HH gratefully acknowledges the financial support of the Kyushu University Innovator Fellowship Program (Quantum Science Area).
HH and AY are deeply grateful for support from JSPS KAKENHI Grant Numbers JP20H01961 and JP22K21345.
YI, KS, and KF are supported by MEXT Quantum Leap Flagship Program (MEXT Q-LEAP) Grant No. JPMXS0120319794 and JST COI-NEXT Grant No. JPMJPF2014.
YI and KS are also supported by JST SPRING Grant No. JPMJSP2138 and the $\Sigma$ Doctoral Futures Research Grant Program from The University of Osaka.
\end{acknowledgments}

\bibliographystyle{apsrev4-2}
\bibliography{apssamp}

\appendix

\section{Proof of Lemma~\ref{lem:nabla_3D}} \label{appendix:lem-nabla_3D}
We prove Lemma~\ref{lem:nabla_3D}.

\noindent
\textit{Proof of Lemma~\ref{lem:nabla_3D}}
\hspace{1em}
We can calculate the right hand side of Eq.~\eqref{eq:nabla-3D} as follows:
\begin{widetext}
\begin{small}
\begin{align}
    \begin{split}
        &\quad 
        \sum_{i=1}^3 \frac{1}{4 \Delta r_i}\left(\frac{1}{\sqrt{n(\mathbf{r}+\Delta r_i\mathbf{e}_i)}}\tilde{u}_i(\mathbf{r}+\Delta r_i\mathbf{e}_i,t) \tilde{\mathbf{u}} (\mathbf{r}+2 \Delta r_i\mathbf{e}_i,t) - \frac{1}{\sqrt{n(\mathbf{r}-\Delta r_i\mathbf{e}_i)}}\tilde{u}_i(\mathbf{r}-\Delta r_i\mathbf{e}_i,t) \tilde{\mathbf{u}} (\mathbf{r}-2 \Delta r_i\mathbf{e}_i,t)\right)
    \end{split} \\
    \begin{split}
        &=\sum_{i=1}^3\frac{1}{4 \Delta r_i} 
        \left\{
        \left( \frac{1}{\sqrt{n}}-\frac{1}{2n\sqrt{n}}\Delta r_i\frac{\partial n}{\partial r_i}+O\left((\Delta r_i)^2\right)
        \right)
        \left( \tilde{u}_i + \Delta r_i \frac{\partial \tilde{u}_i}{\partial r_i} 
        + O\left((\Delta r_i)^2\right) 
        \right)
        \left( \tilde{\mathbf{u}} + 2 \Delta r_i \frac{\partial \tilde{\mathbf{u}}}{\partial r_i} 
        + O\left((\Delta r_i)^2\right) 
        \right)
        \right.
        \\
        &\quad\quad\quad\quad\quad \left.-\left( \frac{1}{\sqrt{n}}+\frac{1}{2n\sqrt{n}}\Delta r_i\frac{\partial n}{\partial r_i}+O\left((\Delta r_i)^2\right)
        \right)
        \left( \tilde{u}_i - \Delta r_i \frac{\partial \tilde{u}_i}{\partial r_i} 
        + O\left((\Delta r_i)^2\right) 
        \right)
        \left( \tilde{\mathbf{u}} - 2 \Delta r_i \frac{\partial \tilde{\mathbf{u}}}{\partial r_i} 
        + O\left((\Delta r_i)^2\right) 
        \right)
        \right\}  
    \end{split} \\
    \begin{split}
        &= \sum_{i=1}^3\left\{
        -\frac{1}{4n\sqrt{n}}\frac{\partial n}{\partial r_i}\tilde{u}_i \tilde{\mathbf{u}}
        + \frac{1}{2\sqrt{n}}\frac{\partial \tilde{u}_i}{\partial r_i}\tilde{\mathbf{u}} + \frac{1}{\sqrt{n}}\tilde{u}_i\frac{\partial \tilde{\mathbf{u}}}{\partial r_i}
        +O\left( \Delta r_i \right) \right\}
    \end{split} \no\\
    &\xrightarrow[\Delta \mathbf{r} \to 0]{} 
    -\frac{1}{4n\sqrt{n}}\left(\tilde{\mathbf{u}}\cdot \left(\nabla n\right) \right) \tilde{\mathbf{u}}
    + \frac{1}{2\sqrt{n}}\left(\nabla\cdot \tilde{\mathbf{u}} \right) \tilde{\mathbf{u}} 
    + \frac{1}{\sqrt{n}}\left(\tilde{\mathbf{u}}\cdot \nabla \right) \tilde{\mathbf{u}}.\label{eq:discretize_ududx}
\end{align}
\end{small}
\end{widetext}
That is, we obtain the statement.
\qed

\section{Relabeling for KvN linearization} \label{appendix:relabeling}

As explained in Sec.~\ref{subsec:KvN-linearization}, 
to realize an exponential reduction in the number of qubits with respect to the number of variables $N$ on a classical computer,
we restrict the total particle number to at most $m$.
In this section,
we describe a method for reducing the number of using qubits 
by relabeling the indices of the state $\otimes_{j=0}^{N-1} \ket{n_j}$.

For relabeling, we introduce the following order $\prec$ on $\left( \mathbb{Z}_{\ge 0}  \right)^N$:
\begin{align}
    &(n_0, \dots, n_{N-1}) \prec (n_0^\prime, \dots, n_{N-1}^\prime) \no \\
    \Longleftrightarrow &\left( \sum_{j=0}^{N-1} n_j < \sum_{j=0}^{N-1} n_j^\prime \right) \quad 
    \mathrm{or} \no \\
    & \Biggl( \sum_{j=0}^{N-1} n_j = \sum_{j=0}^{N-1} n_j^\prime  \quad \mathrm{and} \no \\
    &\exists k \in \{0 \dots, N-1 \}, \ \mathrm{s.t.} \no \\
    & n_k < n_k^\prime, n_l = n_l^\prime \ \left( l \in \{k+1, \dots, N-1\} \right) \Biggr).
\end{align}
This ordering is defined primarily by the total sum of the values in the sequence, 
where sequences with larger sums are considered greater.
If the total sums are equal, 
the order is determined by comparing the values from right to left.
Using the proposed order,
we relabel the indices as follows:
\begin{equation} \label{eq:relabeling}
\begin{split}
    G &\colon 
    \left( \mathbb{Z}_{\ge 0}  \right)^N \rightarrow \mathbb{Z}_{\ge 0} ; \\
    &(n_0, \dots, n_{N-1})\\
    &\mapsto 
    \begin{dcases}
        0 \hspace{3cm}(n_0=\dots=n_{N-1}=0) \\
        \begin{split}
        {}_{N+K-1} C_N + 
        \sum_{k=1}^{N-1}  \sum_{j=0}^{n_k-1} 
        &{}_{K-\sum_{l=k+1}^{N-1}n_l -j+k-1} C_{k-1} \\
        & \hspace{1.6cm} (\mathrm{otherwise})
        \end{split}
    \end{dcases}
    ,
\end{split}
\end{equation}
where $K=\sum_{j=0}^{N-1} n_j$.
This relabeling requires $O({}_{N+m}C_m)=O(N^m)$ variables, 
which implies the number of using qubits is $O(m \log N)$.
This implies an exponential reduction in the number of qubits with respect to the number of variables.

We note that this labeling is different from the labeling utilized in Ref.~\cite{Engel2021, Tanaka2024}.
In Ref.~\cite{Engel2021, Tanaka2024},
any given state $\ket{\mathbf{n}} = \bigotimes_{i=0}^{N-1} \ket{n_i}$ 
is transformed as follows:
\begin{equation}
    \ket{\mathbf{n}} \mapsto \ket{\mathbf{n}_m} 
    \coloneqq
    \ket{
    \underbrace{\overbrace{0 \cdots 0}^{m-\sum_{i=0}^{N-1} n_i} \,
    \overbrace{1 \cdots 1}^{n_0} \,
    \cdots \,
    \overbrace{N \cdots N}^{n_{N-1}}}_m},
\end{equation}
where $0, 1, \dots, N$ are encoded in binary using $\log N$ qubits each.
This encoding is a bit redundant since the mode indices are arranged in ascending and certain basis states are forbidden.
To remove the redundancy, we introduced the proposed labeling.

\section{Block encoding of the mapped Hamiltonian} \label{appendix:block-encoding-hamiltonian}
In this section, we provide a framework for constructing the block encoding of the Hamiltonian $\hat{H}_m$, which is truncated from the Hamiltonian in Eq.~\eqref{eq:def_hamiltonian}. 
We note that the block encoding provided here is intended for implementation on quantum computers in the future and differs from those used in numerical simulations.
Also note that we do not provide the detailed circuit description but the specific direction for constructing the circuit in the future work.

First, we explain the general framework for the block encoding of sparse Hamiltonians.
If we want to encode the Hamiltonian of dimension $M< 2^w$, then we can define an encoded Hamiltonian $A_e\in \mathbb{C}^{2^w\times 2^w}$ such that the top-left block of $A_e$ is $A$ and all other elements are 0.
\begin{lemma}[Hermitian block-encoding of sparse Hamiltonians~\cite{camps2024explicit}]\label{lem:sparse-block-encoding}
Let $A$ be an $s$-sparse Hamiltonian of dimension $2^w$ and $A_{\mathrm{max}}$ be an upper-bound of the max-norm of $A$, i.e., each element of $A$ has absolute value at most $A_{\mathrm{max}}$.
Suppose that we have the following sparse-access oracles acting on two $(w+1)$-qubit registers which query the locations of the non-zero entries:
\begin{equation}
O_C \ket{j}\ket{l} = \ket{j} \ket{c(j,l)},
\end{equation}
where $c(j,l)$ is the index for the $l$-th non-zero entry of the $j$-th column of $A$, or if there are less than $j$ non-zero entries, then it is $j+2^w$.
Also, assume that we have the access to the oracle $O_A$ which returns the entries of $A$:
\begin{equation}
O_A \ket{j}\ket{l} \ket{0} = \ket{j}\ket{l} \left( \sqrt{\frac{A_{c(j,l),j}}{A_{\mathrm{max}}}}\ket{0} + \sqrt{1 - \left|\frac{A_{c(j,l),j}}{A_{\mathrm{max}}}\right|}\ket{1} \right).
\end{equation}
Then we can implement an $(sA_{\mathrm{max}}, w+3, 0)$-block-encoding of $A$ with two uses of $O_C$ and $O_A$ and additionally using $O(w)$ one and two qubit gates.
\end{lemma}
Lemma~\ref{lem:sparse-block-encoding} relies on the existence of the oracle $O_C$.
However, it is non-trivial whether we can implement $O_C$ efficiently.
In our case, we can easily construct an alternative oracle $O_{\tilde{C}}$ acting on $2(w+1)+\lceil\log{s}\rceil$ qubits defined as:
\begin{equation}
O_{\tilde{C}} \ket{j}\ket{l}\ket{0^{w+1}} = \ket{j} \ket{l} \ket{c(j,l)},
\end{equation}
which makes us use an alternative version of block-encoding.
\begin{lemma}[alternative version of hermitian block-encoding of sparse Hamiltonians]\label{lem:alter-sparse-block-encoding}
Let $A$ be an $s$-sparse Hamiltonian of dimension $2^w$ and $A_{\mathrm{max}}$ be an upper-bound of the max-norm of $A$, i.e., each element of $A$ has absolute value at most $A_{\mathrm{max}}$.
Suppose that we have the following sparse-access oracles acting on two $2(w+1)+\lceil\log{s}\rceil$-qubit registers which query the locations of the non-zero entries:
\begin{equation}
O_{\tilde{C}} \ket{j}\ket{l}\ket{0^{w+1}} = \ket{j} \ket{l} \ket{c(j,l)},
\end{equation}
where $c(j,l)$ is the index for the $l$-th non-zero entry of the $j$-th column of $A$, or if there are less than $j$ non-zero entries, then it is $j+2^n$.
Also, assume that we have the access to the oracle $O_A$ which returns the entries of $A$:
\begin{equation}
O_A \ket{j}\ket{l} \ket{0} = \ket{j}\ket{l} \left( \sqrt{\frac{A_{c(j,l),j}}{A_{\mathrm{max}}}}\ket{0} + \sqrt{1 - \left|\frac{A_{c(j,l),j}}{A_{\mathrm{max}}}\right|}\ket{1} \right).
\end{equation}
Hence we can implement an $(s^2 A_{\mathrm{max}}, w+3, 0)$-block-encoding of $A$ with two uses of $O_C$ and $O_A$ and additionally using $O(w)$ one and two qubit gates.
\end{lemma}
\begin{proof}
First, we define the unitary operator $D$ that prepares a uniform superposition state over $s$ computational bases:
\begin{equation}
    D\ket{0^{\lceil \log_2 s \rceil}} = \frac{1}{\sqrt{s}} \sum_{l=0}^{s-1} \ket{l},
\end{equation}
which can be implemented with $O(\lceil \log_2 s \rceil)$ one and two qubit gates~\cite{sanders2020compilation}.
We define the $2(w+1)+2\lceil \log_2 s \rceil + 2$-qubit unitary $V=O_{\tilde{C}} O_A (D\otimes D)$ such that:
\begin{align}
&\ket{0}\ket{0}\ket{0^{\lceil \log_2 s \rceil}}\ket{0^w}\ket{0^{\lceil \log_2 s \rceil}}\ket{i} \\
\stackrel{D\otimes D}{\longrightarrow}& \ket{0}\ket{0} \left( \frac{1}{\sqrt{s}} \sum_{l=0}^{s-1} \ket{l} \right) \ket{0^w} 
\left( \frac{1}{\sqrt{s}} \sum_{k=0}^{s-1} \ket{k} \right) \ket{i} \\
\stackrel{O_{A}}{\longrightarrow} 
\begin{split}
 &\ket{0} \frac{1}{\sqrt{s}} \sum_{l=0}^{s-1} \left( \sqrt{\frac{A_{c(i,l),i}}{A_{\mathrm{max}}}}\ket{0} + \sqrt{1 - \left|\frac{A_{c(i,l),i}}{A_{\mathrm{max}}}\right|} \ket{1} \right)  \\
 &\hspace{4.5cm}  \ket{l} \ket{0^w} \left( \frac{1}{\sqrt{s}} \sum_{k=0}^{s-1} \ket{k} \right) \ket{i} 
\end{split} \\
\stackrel{O_{\tilde{C}}}{\longrightarrow}
\begin{split}
&\ket{0} \frac{1}{\sqrt{s}} \sum_{l=0}^{s-1} \left( \sqrt{\frac{A_{c(i,l),i}}{A_{\mathrm{max}}}}\ket{0} + \sqrt{1 - \left|\frac{A_{c(i,l),i}}{A_{\mathrm{max}}}\right|}\ket{1} \right) \\
&\hspace{4cm} \ket{l} \ket{c(i,l)} \left( \frac{1}{\sqrt{s}} \sum_{k=0}^{s-1} \ket{k} \right) \ket{i} \end{split}\\
\begin{split} 
=& \frac{1}{s} \sum_{l,k=0}^{s} \ket{0} \left( \sqrt{\frac{A_{c(i,l),i}}{A_{\mathrm{max}}}}\ket{0} + \sqrt{1 - \left|\frac{A_{c(i,l),i}}{A_{\mathrm{max}}}\right|}\ket{1} \right) \\
& \hspace{5.5cm} \ket{l} \ket{c(i,l)} \ket{k} \ket{i}.
\end{split}
\end{align}
Applying the $2(w+1)+2\lceil \log_2 s \rceil + 2$-qubit SWAP gate $U_{\mathrm{SWAP}}$, the state becomes
\begin{align}\label{eq:SWAPVinput}
&\frac{1}{s} \sum_{l,k=0}^{s} \left( \sqrt{\frac{A_{c(i,l),i}}{A_{\mathrm{max}}}}\ket{0} + \sqrt{1 - \left|\frac{A_{c(i,l),i}}{A_{\mathrm{max}}}\right|}\ket{1} \right) \no\\
& \hspace{5.5cm}\ket{0} \ket{k} \ket{i} \ket{l} \ket{c(i,l)}.
\end{align}
Then, we can see that $V^\dag U_{\mathrm{SWAP}} V$ is an $(s^2 A_{\mathrm{max}}, w+3, 0)$-block-encoding of $A$ as follows:
Taking the inner product of Eq.~\eqref{eq:SWAPVinput} and 
\begin{equation}
\begin{split}
&\bra{0}\bra{0}\bra{0^{\lceil \log_2 s \rceil}}\bra{0^w}\bra{0^{\lceil \log_2 s \rceil}}\bra{j} V^\dag \\
=& \frac{1}{s} \sum_{l',k'=0}^{s} \bra{0} \left( \sqrt{\frac{A_{c(j,l'),j}^*}{A_{\mathrm{max}}}}\bra{0} + \sqrt{1 - \left|\frac{A_{c(j,l'),j}}{A_{\mathrm{max}}}\right|}\bra{1} \right) \\
& \hspace{5cm} \bra{l'} \bra{c(j,l')} \bra{k'} \bra{j} ,
\end{split}
\end{equation}
we obtain 
\begin{align}
\begin{split}
&\bra{0}\bra{0}\bra{0^{\lceil \log_2 s \rceil}}\bra{0^w}\bra{0^{\lceil \log_2 s \rceil}}\bra{j} \\
&\quad V^\dag U_{\mathrm{SWAP}} V \ket{0}\ket{0}\ket{0^{\lceil \log_2 s \rceil}}\ket{0^w}\ket{0^{\lceil \log_2 s \rceil}}\ket{i} 
\end{split}\\
=& \frac{1}{s^2} \sum_{l',k',l,k,=0}^{s} \frac{\sqrt{A_{c(j,l'),j}^* A_{c(i,l),i} }}{A_{\mathrm{max}}} \delta_{l',k} \delta_{c(j,l'), i} \delta_{k',l} \delta_{j,c(i,l)} \\
=& \frac{1}{s^2} \sum_{l',l=0}^{s} \frac{\sqrt{A_{c(j,l'),j}^* A_{c(i,l),i} }}{A_{\mathrm{max}}} \delta_{c(j,l'), i} \delta_{j,c(i,l)} \\
=& \frac{1}{s^2} \frac{\sqrt{A_{i,j}^* A_{j,i} }}{A_{\mathrm{max}}} 
= \frac{A_{j,i}} {s^2 A_{\mathrm{max}}}.
\end{align}
\end{proof}

By virtue of Lemma~\ref{lem:alter-sparse-block-encoding}, we only have to implement $O_{\tilde{C}}$ and $O_A$ for our Hamiltonian $\hat{H}$ in Eq.~\eqref{eq:def_hamiltonian}.
We describe how to implement $O_{\tilde{C}}$ and $O_A$ simultaneously in the following.

Our target unitary acts as
\begin{equation}
\begin{split}
&O_{\tilde{C}} O_A \left( \ket{j}\frac{1}{\sqrt{mc2^d}}\sum_{l=1}^{mc2^d}\ket{l}\ket{0^{m\log(N)}}\ket{0} \right) \\
=& \ket{j}\frac{1}{\sqrt{mc2^d}}\sum_{l=1}^{mc2^d}\ket{l}\ket{c(j,l)} \\
&\quad\left( \sqrt{\frac{\hat{H}_{c(j,l),j}}{\hat{H}_{\mathrm{max}}}}\ket{0} + \sqrt{1 - \left|\frac{\hat{H}_{c(j,l),j}}{\hat{H}_{\mathrm{max}}}\right|}\ket{1} \right).
\end{split}
\end{equation}
Here, the index $j$ is the compressed form of some particle configuration $\bm{n_j}$ transformed by the function $j=G(\bm{n_j})$ defined in Eq.~\eqref{eq:relabeling}.
The index $c(j,l)$ is the location of $l$-th nonzero entry, which is also the compressed form of some particle configuration $\bm{n_{c(j,l)}}$ satisfying $\braket{\bm{n_{c(j,l)}}|\hat{H}|\bm{n_j}}\neq 0$.
Thus, we need to extract information of particle configuration from input register $\ket{j}$ and then calculate all non-zero entries and their locations.
Recall that our Hamiltonian $\hat{H}$ has the sparsity of $mc2^d$.
This comes from the three facts that, for some configuration of particles, 
the number of sites where there exists at least one particle is at most $m$, that, for each site, 
there exist at most $c$ terms of Hamiltonian which act on the site, and that, for each term of Hamiltonian, the particle number changes in at most $2^d$ ways by at most $d$ creation and annihilation operators.
This enables us to interpret the input uniform superposition state as:
\begin{equation}
\ket{j}\frac{1}{\sqrt{mc2^d}}\sum_{k=1}^{m}\sum_{l=1}^{c}\sum_{\delta_1,\dots,\delta_d = \{-1,1\}}\ket{k,l,\delta_1,\dots,\delta_d},
\end{equation}
where $k$ denotes the site where particles exist, $l$ denotes the term of Hamiltonian which act on the site,
and $\delta_1,\dots,\delta_d$ represent the creation and annihilation of particles.
Using this notation and some oracles provided below, we can implement the target unitary as follows:
\begin{enumerate}
    \item Oracle-1: Add $\log(N) + \log(m)$ ancilla qubits and calculate the site $\gamma_k(j)$ holding particles and its particle number $n_k(j)$ conditioned on $j$ and $k$:
    \begin{equation}
    \begin{split}
        &\ket{j}\frac{1}{\sqrt{mc2^d}}\sum_{k=1}^{m}\sum_{l=1}^{c}\sum_{\delta_1,\dots,\delta_d = \{-1,1\}}\ket{k,l,\delta_1,\dots,\delta_d} \\
        &\hspace{5.2cm}\ket{\gamma_k(j)} \ket{n_k(j)}.
    \end{split}
    \end{equation}
    \item Oracle-2: Add $d(\log(N) + \log(m))$ ancilla qubits and, conditioned on $\gamma_k(j)$ and $l$, calculate $d$ sites where the Hamiltonian acts on:
    \begin{equation}
    \begin{split}
        &\ket{j}\frac{1}{\sqrt{mc2^d}}\sum_{k=1}^{m}\sum_{l=1}^{c}\sum_{\delta_1,\dots,\delta_d = \{-1,1\}}\ket{k,l,\delta_1,\dots,\delta_d} \\
        &\ \ket{\gamma_k(j)} \ket{n_k(j)} \ket{\gamma_{k,l,1}}\dots\ket{\gamma_{k,l,d}} \ket{n_{k,l,1}}\dots\ket{n_{k,l,d}}.
    \end{split}
    \end{equation}
    Here, when $H_{k,l}$ denotes $l$-th term among the Hamiltonian terms which act on site $\gamma_k(j)$, $\gamma_{k,l,1} \dots \gamma_{k,l,d}$ denote the sites on which $H_{k,l}$ acts and $n_{k,l,1} \dots n_{k,l,d}$ are their particle number.
    \item Conditioned on $\delta_1,\dots,\delta_d$, transform the register storing the particle number as:
    \begin{equation}
    \begin{split}
        &\ket{j}\frac{1}{\sqrt{mc2^d}}\sum_{k=1}^{m}\sum_{l=1}^{c}\sum_{\delta_1,\dots,\delta_d = \{-1,1\}}\ket{k,l,\delta_1,\dots,\delta_d} \\
        &\hspace{2.5cm}\ket{\gamma_k(j)} \ket{n_k(j)} \ket{\gamma_{k,l,1}}\dots\ket{\gamma_{k,l,d}} \\
        &\hspace{3.3cm}\ket{n_{k,l,1}+\delta_1}\dots\ket{n_{k,l,d}+\delta_d}.
    \end{split}
    \end{equation}
    This operation can be done by a controlled version of an addition operator using $O(2^d \log(m))$ gates~\cite{gidney2018halving, berry2019qubitization}.
    \item Oracle-3: Conditioned on $\gamma_k(j)$, $l$ and $\ket{n_{k,l,1}+\delta_1}\dots\ket{n_{k,l,d}+\delta_d}$, calculate the entry of  the Hamiltonian and rotate the single qubit:
    \begin{equation}
    \begin{split}
        &\ket{j}\frac{1}{\sqrt{mc2^d}}\sum_{k=1}^{m}\sum_{l=1}^{c}\sum_{\delta_1,\dots,\delta_d = \{-1,1\}}\ket{k,l,\delta_1,\dots,\delta_d} \\
        &\ket{\gamma_k(j)} \ket{n_k(j)} \ket{\gamma_{k,l,1}}\dots\ket{\gamma_{k,l,d}} \ket{n_{k,l,1}+\delta_1}\dots\ket{n_{k,l,d}+\delta_d}\\
        &\left( \sqrt{\frac{\hat{H}_{c(j,k,l,\delta_1,\dots,\delta_d),j}}{\hat{H}_{\mathrm{max}}}}\ket{0} + \sqrt{1 - \left|\frac{\hat{H}_{c(j,k,l,\delta_1,\dots,\delta_d),j}}{\hat{H}_{\mathrm{max}}}\right|}\ket{1} \right).
    \end{split}
    \end{equation}
    \item Oracle-4: Conditioned on $j$ and $\ket{\gamma_{k,l,1}} \dots \ket{\gamma_{k,l,d}} \ket{n_{k,l,1}+\delta_1} \dots \ket{n_{k,l,d}+\delta_d}$, calculate the index $c(j,k,l,\delta_1,\dots,\delta_d)$ corresponding to the particle configuration after the Hamiltonian is applied.
    \begin{equation}
    \begin{split}
        &\ket{j}\frac{1}{\sqrt{mc2^d}}\sum_{k=1}^{m}\sum_{l=1}^{c}\sum_{\delta_1,\dots,\delta_d = \{-1,1\}}\ket{k,l,\delta_1,\dots,\delta_d} \\
        &\hspace{0.5cm}\ket{\gamma_k(j)} \ket{n_k(j)} \ket{\gamma_{k,l,1}}\dots\ket{\gamma_{k,l,d}} \\
        &\hspace{1cm} \ket{n_{k,l,1}+\delta_1}\dots\ket{n_{k,l,d}+\delta_d} \ket{c(j,k,l,\delta_1,\dots,\delta_d)} \\
        &\left( \sqrt{\frac{\hat{H}_{c(j,k,l,\delta_1,\dots,\delta_d),j}}{\hat{H}_{\mathrm{max}}}}\ket{0} + \sqrt{1 - \left|\frac{\hat{H}_{c(j,k,l,\delta_1,\dots,\delta_d),j}}{\hat{H}_{\mathrm{max}}}\right|}\ket{1} \right).
    \end{split}
    \end{equation}
    \item Uncomputing step 1, 2 and 3, we obtain the desired state
    \begin{equation}
    \begin{split}
        &\ket{j}\frac{1}{\sqrt{mc2^d}}\sum_{k=1}^{m}\sum_{l=1}^{c}\sum_{\delta_1,\dots,\delta_d = \{-1,1\}}\\
        &\quad\ket{k,l,\delta_1,\dots,\delta_d}
        \ket{c(j,k,l,\delta_1,\dots,\delta_d)} \\
        &\left( \sqrt{\frac{\hat{H}_{c(j,k,l,\delta_1,\dots,\delta_d),j}}{\hat{H}_{\mathrm{max}}}}\ket{0} + \sqrt{1 - \left|\frac{\hat{H}_{c(j,k,l,\delta_1,\dots,\delta_d),j}}{\hat{H}_{\mathrm{max}}}\right|}\ket{1} \right).
    \end{split}
    \end{equation}
\end{enumerate}

Our construction is based on the sparse-access model rather than other methods such as linear combination of unitaries technique.
Operationally, we can also implement the block-encoding by taking a linear combination of the block-encoding of each term $\hat{k}_j \prod_{l \in p \setminus \{j\}} \hat{x}_l$.
However, this implementation is inefficient since the Hamiltonian has $\mathrm{poly}(N)$ terms, which leads to the gate complexity of $\mathrm{poly}(N)$.

\end{document}